\documentclass[12pt]{article}
\usepackage[final]{epsfig}
\usepackage{amssymb}
\usepackage{latexsym}
\usepackage{wrapfig}
\usepackage{amsmath}
\usepackage{booktabs}
\usepackage{threeparttable}
\usepackage{xcolor}
\usepackage{amsthm}
\usepackage{comment}
\usepackage{graphicx, caption, subcaption}

\newcommand{\R}{\mathbb{R}}
\newcommand{\Z}{\mathbb{Z}}
\newcommand{\bfa}{{\bf a}}

\newcommand{\bfc}{{\bf c}}

\newcommand{\bfe}{{\bf e}}
\newcommand{\bff}{{\bf f}}
\newcommand{\bfg}{{\bf g}}
\newcommand{\bfh}{{\bf h}}
\newcommand{\bfk}{{\bf k}}

\newcommand{\bfn}{{\bf n}}

\newcommand{\bfp}{{\bf p}}

\newcommand{\bfr}{{\bf r}}

\newcommand{\bft}{{\bf t}}
\newcommand{\bfu}{{\bf u}}
\newcommand{\bfv}{{\bf v}}

\newcommand{\bfx}{{\bf x}}
\newcommand{\bfy}{{\bf y}}

\newcommand{\bfz}{{\bf z}}
\newcommand{\bfA}{{\bf A}}

\newcommand{\bfI}{{\bf I}}

\newcommand{\bfQ}{{\bf Q}}
\newcommand{\bfR}{{\bf R}}

\newcommand{\bfW}{{\bf W}}

\newcommand{\beq}{\begin{equation}}
\newcommand{\eeq}{\end{equation}}
\newcommand{\beqs}{\begin{eqnarray}}
\newcommand{\eeqs}{\end{eqnarray}}

\newcommand{\calC}{{\cal C}}
\newcommand{\calD}{{\cal D}}

\newcommand{\calU}{{\cal U}}

\newtheorem{theorem}{Theorem}[section]

\newtheorem{lemma}{Lemma}[section]

\DeclareMathOperator{\sign}{sign}

\begin{document}

\begin{center}
\Huge

{Phase transformations and compatibility in helical structures} \\

\vspace{5mm}
\large

\vspace{2mm}
Fan Feng, Paul Plucinsky and Richard D. James  \\
Department of Aerospace Engineering and Mechanics \\
University of Minnesota  \\
{\small fengx258@umn.edu, pplucins@umn.edu,  james@umn.edu}\\
\end{center}

\vspace{0.2in}
\normalsize

\noindent {\bf Abstract.}   
We systematically study phase transformations from one helical structure to another.
Motivated in part by recent work that relates the presence of compatible interfaces with 
properties such as the hysteresis and reversibility of a phase transformation \cite{zarnetta_10,song_enhanced_2013,chluba_15,ni_16}, we give
necessary and sufficient conditions on the structural parameters of  two helical phases such
that they are compatible.  We show that, locally, four types of compatible interface are possible:  
vertical, horizontal, helical and elliptical.  We discuss the mobility of these interfaces and give
examples of systems of interfaces that are mobile and could be used to fully transform a helical
structure from one phase to another.

These results provide a basis for the tuning of helical
structural parameters so as to achieve compatibility of phases.  In the case of transformations in crystals, this kind 
of tuning has led to materials with exceptionally low hysteresis and dramatically improved resistance to transformational 
fatigue.  Compatible helical transformations with low hysteresis and fatigue resistance would exhibit
an unusual shape memory effect involving both twist and extension, and may have potential applications 
as new artificial muscles and actuators.

\vspace{0.2in}

 \vspace{0.2in}
\tableofcontents

\vspace{0.4in}

\section{Introduction} \label{sect1}
A {\it helical structure} is a molecular structure obtained by taking the orbit of a single molecule 
under a \textit{helical group}.  The helical groups are the discrete groups of isometries (i.e., orthogonal transformations and translations)
that do not contain any pure translations and do not fix a point.  Helical structures are special cases of the more general concept of 
{\it objective structures} \cite{james_objective_2006}.  An objective structure is a discrete collection of atoms where corresponding atoms
in each molecule  ``see the same environment''.  For reasons that are not understood and are related
to the celebrated and far-from-solved ``crystallization problem'', objective structures are surprisingly well-represented among scientifically 
and technologically important nanostructures. 

Helical structures, in particular, are ubiquitous. Single-walled carbon nanotubes of any chirality as well as many biological molecules are helical.
They are especially common in non-animal viruses. 
An example of a helical structure that undergoes a phase transformation of the type discussed here is the tail sheath of bacteriophage T4 \cite{moody_67,moody_73}.  
This highly effective transformation is employed
by the T4 virus to drive the stiff tail tube through the cell wall of the host, and the viral DNA then enters the host through
the tail tube.   It can be seen that, aside from being helical, this transformation
has many features in common with martensitic phase transformations in crystals as first noticed
by Olson and Hartman \cite{olson_82},
such as being diffusionless,  and exhibiting a large coordinated shape change.  The
transformation also exhibits a latent heat \cite{arisaka_81}, and the transformation can be accurately modeled by a free energy
of the type used in studies of phase transformations in crystals \cite{falk_elasticity_2006}.
The forms of the helical generators  of both phases of this tail tube are special cases of those studied here,
but the phases of bacteriophage T4 do not precisely satisfy the strong conditions of compatibility found here. 
This is possibly related
to the fitness requirement of T4 to have sufficiently large hysteresis so that the transformation does not occur spontaneously.  
In fact, T4 has a trigger involving its tail fibers and baseplate which induces the transformation
only after it has attached itself to its host. 

Another typical example is the phase transformation of bacterial flagella, which are mechanically induced by the bacterium's molecular motor \cite{asakura_70,calladine_75, yonekura_03}. The bacterium can enter into ``swimming" mode to move toward a favorable chemical and thermal environment, 
or a ``tumbling" mode to alter its direction, by switching chirality of the flagellum.   The thermodynamics of such phenomena was explained by minimizing Gibbs free energy, 
including chemical and mechanical parts,  and thermomechanical phase diagrams are given in \cite{ricardo_2015}.  Other examples analyzed in Section \ref{sec9} are certain microtubules satisfying closely the compatibility conditions derived here.  Some widely studied nonbiological helical structures are the celebrated
carbon nanotubes (any chirality) and the nanotubes 
BCN  \cite{zhou_14}, GaN \cite{goldberger_03}, MoS$_2$ and WS$_{2}$ \cite{rad_11,zhang_10}.  

A recent development that can guide the of tuning of lattice parameters is an implementation
of density functional theory that incorporates helical symmetry, so that typical total energy calculations on nanotubes with no axial periodicity 
can be carried out with a few atom calculation\footnote{two atoms for twisted/extended chiral C nanotubes.} \cite{banerjee_spectral_2015,amartya_16,amartya_18,Amartya_PhD_Thesis}.
With these tools many different  twisted and extended nanotubes can be evaluated and phase transformations can be identified.  Twist and extension 
are not only valuable parameters to seek novel phase transformations, but they can also be used to seek special conditions of compatibility identified here.
Phase transformations having a change in magnetoelectric or transport properties  \cite{amartya_18} are particularly interesting in an engineering context 
due to the wire-like geometry of nanotubes,
together with the fact that their lengths can be macroscopic.

In this paper we develop a theory of diffusionless phase transformations in helical structures with a focus on low energy interfaces. 
The main guideline behind this study is to systematically replace the translation group by the helical group, but to otherwise exploit the 
patterns of thought used in atomistic and continuum theories of phase transformations \cite{ball_fine_1987, james_martensitic_2000, james_taming_2015, 
song_enhanced_2013, bhattacharya_self-accommodation_1992, ball_proposed_1992}.   
We show that familiar concepts from phase transformations in crystals such as {\it variants}, {\it twins}, {\it compatible interfaces}, {\it slip planes}  and {\it habit planes} have analogs 
in the helical case, though the analogy is not perfect.
 
To develop a phase transformation theory for helical structures, we consider structures generated by the
the largest Abelian discrete helical group acting on a finite set of points in $\mathbb{R}^3$, as described 
in Section \ref{sec2}.  This assumption includes structures generated by all the helical groups as long as we choose the 
generating molecule appropriately, as we 
explain below.  We then focus on compatible interfaces.   Compatibility is fundamentally a metric property, i.e., it concerns
conditions under which atoms that are close together before transformation remain close together after
transformation.  However, closeness of powers of group generators does not generally imply closeness of
molecules.  Therefore, we are led to develop a reparameterization of the group by {\it nearest neighbor generators}, as 
well as a suitable domain\textemdash analogous to a reference configuration of continuum mechanics\textemdash of
powers of the generators that imply a 1--1 relation between powers and molecules with 
convenient metric properties.  The reparameterization of the groups in terms of nearest neighbor generators and the
characterization of their domains is presented as a series of rigorously derived algorithms in Sections \ref{sec3} and \ref{sec4}.

The powers of generators are integers, but we notice that the resulting formulas for molecular positions
make sense for non-integer values of the powers and still the metric properties hold: closeness of powers
(whether integer or not) implies closeness of molecules.  Using  this observation, we then define compatibility
in continuum mechanics terms (cf., \cite{zigzag_16}) in Section \ref{sec5}.  Further, we work out all local solutions for compatible interfaces 
for all values of the group parameters (under mild restrictions) in Section \ref{sec6} and we give explicit closed form formulas for all compatible
interfaces.   The conditions on group parameters for compatible interfaces that 
we find are strictly analogous
to the condition $\lambda_2 = 1$ in theories of phase transformations in crystals \cite{james_05,cui_06,zarnetta_10,zjm}.  We conjecture that the dramatic reduction of the sizes of hysteresis loops
observed in transforming crystals when $\lambda_2$ is tuned to the value $1$ by compositional changes will also occur in helical structures when
the conditions found here are satisfied.  

Some of the local interfaces we find can all be extended to form loops or infinite lines (Section \ref{sect6.2}).   Moreover, one can combine various 
compatible interfaces to form nanostructures 
that are analogous to supercompatible interfaces in crystals \cite{chen_study_2013,gu_17}.  In particular in Section \ref{sec9} 
we find an interesting helical analog
of an austenite/martensite interface seen in supercompatible bulk crystals \cite{chen_study_2013,song_enhanced_2013} : it is moveable without slip and involves twinning, exhibits no stressed transition layer, but otherwise looks very different from usual austenite/martensite interfaces. 

An important special case is that in which the two phases are the same, i.e., the two helical lattices are related by an orthogonal transformation and translation.
In this case the compatible interfaces we find are  naturally interpreted as helical analogs of slips or twins.  We study this  case in Sections \ref{sec7} and \ref{sec8}.  Again, some
exact analogs with the crystalline case (e.g., mirror symmetry) emerge, but there are also cases that have no crystalline analog.  

It should be noted that interfaces of the type identified here in atomic structures are also seen in macroscopic hollow tubes
made of NiTi shape memory material \cite{sun_09}.  For related theory at macroscopic level see also \cite{zigzag_16}.

\section{Isometry groups and helical structures} \label{sec2}

As noted  in the introduction, a helical structure is the orbit of a molecule under a helical group.   
To define this precisely, let $\bfp_1, \dots, \bfp_M$ be the average positions of atoms in a  molecule  (with
corresponding species).  A helical group can be represented
schematically by $\{g_0, g_1, g_2,\dots \}$ with say $g_0 = identity$.  Consequently, a helical structure is given by $\{g_0(\bfp_1), \dots, g_0(\bfp_M) \} \cup
\{g_1(\bfp_1), \dots, g_1(\bfp_M) \} \cup \{g_2(\bfp_1), \dots, g_2(\bfp_M) \} \cup \dots $.

A helical group is a discrete group consisting of {\it isometries} that do not fix a point and which does not contain any pure translations.  
That is, each $g_i, \ i = 0, 1, 2, \dots$, is an isometry of the form
$(\bfQ_i | \bfc_i)$ in conventional notation, where $\bfQ_i$ is an orthogonal transformation on $\mathbb R^3$
and $\bfc_i \in \mathbb R^3$.  The condition that group contains no translation is
 $(\mathbf{Q}_i|\mathbf{c}_i) \neq (\mathbf{I}|\mathbf{c})$ for any $\mathbf{c} \neq 0$, and the condition that the group does not fix a point is
that there is no $\bfx_0 \in \mathbb R^3$ such that $(\bfQ_i - \bfI)\bfx_0 + \bfc_i = 0$ for every $i$.  (If there is such an $\bfx_0$, the resulting group is a point group.) 
The group product is $g_i g_j  = (\bfQ_i \bfQ_j | \bfQ_i \bfc_j + \bfc_i)$,
the identity is $g_0 = (\bfI | 0)$, and the inverses are $g_i^{-1}  = (\bfQ_i^T | -\bfQ_i^T \bfc_i)$.  
Further, the action of $g_i$ on $\bfp \in \mathbb R^3$, as indicated above,
is simply $g_i (\bfp) =
\bfQ_i \bfp + \bfc_i$.  Finally, the {\it orbit} of $\bfp$ is the collection $g_0(\bfp), g_1(\bfp), g_2 (\bfp), \dots$

Volume E of the International Tables of Crystallography (IT) contains a listing of {\it subperiodic groups}, i.e., the discrete isometry groups not containing a full set of 3 linearly independent translations. By definition, helical groups are discrete isometry groups containing no translations, and so they would appear to fall under the umbrella of this classification.  However, as is known to crystallographers, Volume E of IT does not contain the helical groups.  This is due to an unfortunate feature of the scheme by which IT is organized.  That is, in IT two isometry groups $G_1$ and $G_2$ are considered the same if they are related by an affine transformation,
$G_2  = a G_1 a^{-1}$ where $a = (\bfA | \bfc)$, $\det \bfA \ne 0$ (or, for some parts of IT, $\det \bfA >0$).
Here the product rule is the same as the one given above\footnote{For $\bfQ_i^T$ substitute $\bfA^{-1}$}.
By this classification, there are infinitely many helical groups and a listing according to the scheme of IT is impossible.  For example, by this classification scheme, the simple helical group specified in (\ref{G1}) below is actually an infinite number of different groups, one for each distinct choice of the angle\footnote{that is,  there is no affine transformation which relates 
$\{ h^m  : m \in \Z \}$ and $\{\tilde{h}^m : m \in \Z \}$ for $\theta \neq \tilde{\theta}.$ (see (\ref{G1}) for notation).}  $\theta$.

For the purpose of this paper, and, one could argue, for many other purposes in science and engineering, the
affine equivalence is not relevant, as one would often like to know ``what are all the groups".  Indeed, for the purpose of
exploring conditions of compatibility, we need explicit formulas for the groups with all the free parameters displayed\footnote{Abstract groups (multiplication tables) are not so useful, and for the applications in this paper it does
not matter if the abstract group suddenly gets bigger at a particular set of values of the parameters.}.   
We have rigorously derived the formulas for all the helical groups in this way \cite{formulas_18}. 
From this, every helical group is given by one of four formulas:
\beqs
& & \{ h^m  : m \in \Z \},   \label{G1} \\
& &  \{ h^m f^{s} : m \in \Z,\, s = 1,2  \},   \label{G2}  \\
& & \{ h^m g^n : m \in \Z,\, n = 1, \dots, i  \},    \label{G3}  \\
& &  \{ h^m g^n f^{s} : m \in \Z,\, n = 1, \dots, i,\, s = 1,2  \},  \label{G4} 
\eeqs
where
\begin{enumerate}
\item  $h = (\bfQ_{\theta}| \tau \bfe + (\bfI - \bfQ_{\theta})\bfz \}$,
$\bfQ_{\theta} \bfe = \bfe,\, |\bfe| =1,\, \bfz \in \mathbb R^3,\, \tau \in \mathbb{R}\setminus \{ 0\}$,
is a screw displacement with an  angle $\theta$ that is an irrational multiple of $2 \pi$.
\item  $g = (\bfQ_{\alpha}| (\bfI-\bfQ_{\alpha}) \bfz )$,
$\bfQ_{\alpha} \bfe = \bfe$,
is a proper rotation with  angle $\alpha = 2\pi/i,\, i \in \mathbb{N}, \, i \ne 0$.
\item $f = (\bfQ|\,
(\bfI-\bfQ) \bfz_1),\, \bfQ =-\bfI + 2 \bfe_1 \otimes \bfe_1,
\, |\bfe_1| = 1,
\bfe \cdot \bfe_1 = 0$ is a 180$^\circ$ rotation with
axis perpendicular to $\bfe$.   Here, $\bfz_1 = \bfz  + \xi \bfe$,
for some $\xi \in \R$.
\end{enumerate}

\begin{figure*}[ht!]
	\centering
		\includegraphics[width=\textwidth]{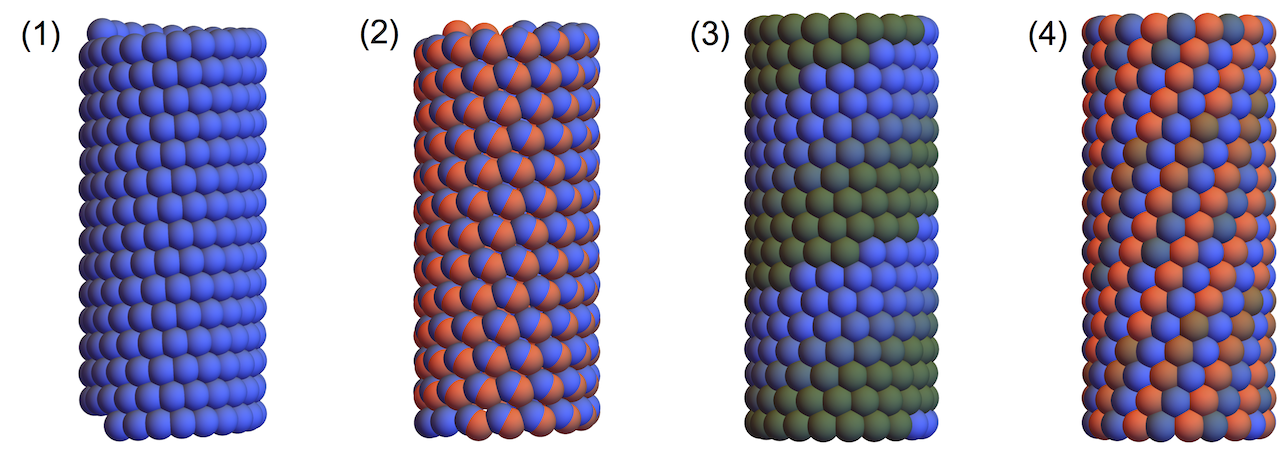}
	\caption{Helical groups corresponding to (\ref{G1}), (\ref{G2}), (\ref{G3}), (\ref{G4}), respectively.  Each picture is
	the orbit of a single ball under the corresponding group and the coloring is according to the powers $s$ or $n$. }
	\label{fig:helical_groups}
\end{figure*}

Groups (\ref{G1}) and (\ref{G3}) are Abelian (the elements commute) and clearly  (\ref{G1}) is
a subgroup of  (\ref{G3})\footnote{Indeed, if we restrict our attention to the case $i = 1$ for the groups in (\ref{G3}), then $g = id$ and we obtain the groups in (\ref{G1}).}.  Groups (\ref{G2}) and (\ref{G4}) are not Abelian because
the element $f$ does not commute with $g$ or $h$.  However, note that $f^2 = identity$.
Thus, the action of $f$ on a suitable molecule produces a nearby molecule.  Then, to get
the full structure in the orbit of (\ref{G4}), we operate the group  (\ref{G3}) on this pair of
molecules.  If the structure is reasonably dense with molecules, we can always consider this pair to be
close together, since there will be some molecule close to the axis through $\bfe_1$.
(Note that $\bfe_1$ is perpendicular to the axis of the cylinder on which the structure lies.)
But, if this pair is close together, then we can reasonably discuss compatibility in terms
of the group  (\ref{G3}).  We therefore focus on the largest collection of Abelian helical groups  (\ref{G3}) below, as this contains all possible Abelian groups that can generate a helical structure.

Another reason for this focus is based on discrete notions of compatibility studied in \cite{falk_elasticity_2006} 
which suggest that compatibility is related fundamentally to Abelian groups.  This concerns the interpretation of 
compatibility as relating to the process of returning to the same atom position by going around a loop
in the space $\mathbb Z^2$ of powers of the group elements.

We note that incidentally our assumptions also include several rod groups.  (Rod groups have periodicity
along the axis through $\bfe$.)  This is because we do not use
the assumption that $\theta$ is an irrational multiple of $2\pi$ in any of the results below.

Consider the two phases $a$ and $b$ each described as orbits of a molecule under the
general Abelian helical group (\ref{G3}). 
Assign group parameters
$\theta_{a,b} \in \mathbb{R}$,
$\alpha_{a,b} = 2 \pi/i_{a,b},\ i_{a,b} \in \mathbb{N}\setminus \{0\}$,
$\tau_{a,b} \in \R$, $\bfz_{a,b} \in \R^3$ and $\bfe_{a,b} \in \R^3,\
|\bfe_{a,b}| = 1$.  Define $\bfQ^{a,b}_{\xi}
\in$ SO(3) having axis $\bfe_{a,b}$ and angle $\xi$, i.e.,
\beq \label{getQ}
 \bfQ^{a, b}_{\xi} = \sin \xi \bfW^{a,b} + \cos{\xi} (\bfI - \bfe_{a,b} \otimes \bfe_{a,b}) + \bfe_{a,b} \otimes \bfe_{a,b},
\eeq
respectively, where $\bfW^{a,b} = -\bfW^{a,b\ T} = 
-\bfe_1^{a,b} \otimes \bfe_2^{a,b} + \bfe_2^{a,b} \otimes \bfe_1^{a,b}$, where $\bfe_1^{a,b}, \bfe_2^{a,b}, \bfe_{a,b}$
is a right-
 orthonormal basis.
From this hypothesis\textemdash in particular that $\bfe_a, \bfe_b$ do not depend
on $\xi$\textemdash we also have that
\beq
\frac{d}{d\xi} \bfQ^a_{\xi} = \bfQ^a_{\xi} \bfW^a, \quad
\frac{d}{d\xi} \bfQ^b_{\xi} = \bfQ^b_{\xi} \bfW^b.
\eeq
Note that $\bfQ_{\xi}^{a,b}\bfW^{a,b}=\bfW^{a,b}\bfQ_{\xi}^{a,b}$, respectively.
The two elements $h^m$ and $g^n$ can be combined
and the group (\ref{G3}) can be written as
\beqs
G_a &=& \left\{ \bigg( \bfQ^{a}_{m \theta_{a} + n \alpha_{a}} \bigg|\
m \tau_a \bfe_a + (\bfI -  \bfQ^{a}_{m \theta_{a} + n \alpha_{a}})\bfz_a
 \bigg): m \in \Z,\ n = 1,\ldots, i_a   \right\}, \nonumber \\
G_b &=& \left\{ \bigg(  \bfQ^{b}_{m \theta_{b} + n \alpha_{b}}\, \bigg|\
m \tau_b \bfe_b +\, (\bfI -  \bfQ^{b}_{m \theta_{b} + n \alpha_{b}})\bfz_b
\bigg):  m \in \Z,\ \, n =1,\ldots, i_b  \right\},  \label{groups}
\eeqs
with sub/superscripts $a$ and $b$ denoting the two phases.
It can be seen from the formulas in (\ref{groups}) that $G_{a,b}$ are
groups under the product rule for isometries given in the introduction, e.g.,
\beqs
\lefteqn{\bigg( \bfQ^{a}_{m \theta_{a} + n \alpha_{a}} \bigg|\
m \tau_a \bfe_a + (\bfI -  \bfQ^{a}_{m \theta_{a} + n \alpha_{a}})\bfz_a
 \bigg)\bigg( \bfQ^{a}_{m' \theta_{a} + n' \alpha_{a}} \bigg|\
m' \tau_a \bfe_a + (\bfI -  \bfQ^{a}_{m' \theta_{a} + n' \alpha_{a}})\bfz_a
 \bigg)}  \nonumber \\
& &   = \bigg( \bfQ^{a}_{(m+m') \theta_{a} + (n+ n') \alpha_{a}} \bigg|\
(m+m') \tau_a \bfe_a + (\bfI -  \bfQ^{a}_{(m+m') \theta_{a} + (n+n') \alpha_{a}})\bfz_a
 \bigg). \label{ver}
\eeqs
Note that if $(n+n')> i_a$, then $(n+n')$ in (\ref{ver}) can be replaced by
$(n+n')\ {\rm mod}\ i_a$.

The two groups $G_a$ and $G_b$ have different parameters.  Our main task is to determine
all choices of these parameters that give compatible interfaces.

These groups act on position vectors of atoms in one molecule as described above.  For the purpose of studying compatibility
we picture the molecule as reasonably compact and we discuss conditions of compatibility in terms of
its center of mass.  The precise statement of this assumption is that a typical diameter of the molecule is
on the order of, or less than, the nearest neighbor distance between molecules defined in Section \ref{sec3}.  The
examples of helical structures given in the introduction (including the tail sheath of bacteriophage T4) have
this property.   

The helical structures $a$ and $b$ consist of
atomic positions given by the corresponding
groups each acting on its respective center-of-mass position $\bfp_{a,b}$. 
Therefore, the  center-of-mass positions of the helical structures 
are given by 
\beqs
\bfy_a(n,m) &=& \bfQ^{a}_{m \theta_{a} + n \alpha_{a}}(\bfp_a - \bfz_a) +
m \tau_a \bfe_a + \bfz_a,\ m\in \mathbb{Z},\ n= 1, \ldots, i_a, \nonumber \\
\bfy_b(n,m) &=& \bfQ^{b}_{m \theta_{b} + n \alpha_{b}}(\bfp_b - \bfz_b) +
m \tau_b \bfe_b +\bfz_b,\ m\in \mathbb{Z},\ n= 1,\dots, i_b.
\eeqs
In the language of objective structures the structure $a$ as viewed from the center of 
mass position $\bfy_a(m,n)$ is exactly the same as the structure viewed from $\bfy_a(m',n')$,
for any choices of the integers $m,n,m',n'$ (even though there may be no atoms at these centers of mass).

\section{Nearest-neighbor reparameterization of the groups}\label{sec3}

In this section, we drop the superscripts $a,b$ and consider a single helical group $G$ defined 
as above by
\beq
G =   \left\{ g(n,m) : m \in \Z,\ n = -i,\ldots,0,\dots, i  \right\},
\quad g(n,m) = \bigg( \bfQ_{m \theta + n \alpha} \bigg|\
m \tau \bfe + (\bfI -  \bfQ_{m \theta + n \alpha})\bfz
 \bigg).  \label{G}
\eeq
Here, we have extended the domain of integers for $n$ to include $-i, -i + 1, \ldots, 0$ for technical reasons.  All this does is simply count some group elements more than once, which does not change $G$ (the collection of all such elements).
As before, let the atomic positions be given by $\bfy(n,m) = g(n,m)(\bfp)$.
We assume without loss of generality that, $\bfp-\bfz \neq 0$, $(\bfp - \bfz) \cdot \bfe = 0$, $\tau \ne 0$ 
to avoid degenerate structures (lines, rings).   Note
 that $g(0,0) = id$, so that 
$\bfy(0,0) = \bfp$. 

Conditions of compatibility between phases ensure that nearby atoms before 
transformation remain near each other after transformation.  Thus,  distances are
important.  However, the standard 
parameterization of the groups given above
in terms of $n$ and $m$ does not in general have the property that
if $\bfy(n,m)$ is near $\bfy(n',m')$ in $\mathbb R^3$, then $(n,m)$ is near $(n',m')$
in $\mathbb Z^2$.  Therefore,
it is desirable to reparameterize the groups so that nearest and next-to-nearest 
neighbors of any point $\tilde{\bfy}(n,m)$ are
$\tilde{\bfy}(n+1,m)$, $\tilde{\bfy}(n,m+1)$.  Here $\tilde{\bfy}(n, m)$ is the deformation induced by a new parameterization of the group. Because $G$ is an isometry group (i.e., preserves distances),
we then 
have at least four nearest and next-to-nearest neighbors with positions
$\tilde{\bfy}(n+1,m)$, $\tilde{\bfy}(n-1,m)$, $\tilde{\bfy}(n,m+1)$, $\tilde{\bfy}(n,m-1)$. 
(Of course, there may be additional nearest, or next-to-nearest,
neighbors such as the case when $\tilde{\bfy}(n,m)$ is surrounded by six nearest
neighbors.)   As we show below, under mild assumptions on the group parameters,
it is always possible to find such nearest neighbor generators.

A nearest neighbor reparameterization implies that nearest (and next to 
nearest) neighbors of the reparameterized structure correspond to nearest neighbors of $(0,0)$ 
in the 2D lattice $\mathbb Z^2$.  Thus, we consider
\begin{equation}
\begin{aligned}
dist^2(n,m)  = |\bfy(n,m) - \bfy(0,0)|^2 &= |(\bfQ_{m \theta + n \alpha} - \bfI) (\bfp - \bfz) + m \tau \bfe |^2\\
&= 4r^2 \sin^2\! \bigg(\frac{m \theta + n \alpha}{2}\bigg)  + m^2 \tau^2,  \label{dist}
\end{aligned}
\end{equation}
where $r = |\bfp - \bfz|$ (and subject to appropriate constraints).  

Nearest and second nearest neighbors are obtained by minimizing\footnote{In this formula, we have chosen the reference atom $\mathbf{y}(0,0)$ simply for convenience.  Notice that the distance from the reference atom to its nearest and next nearest neighbors  is independent of the particular choice of reference atom.  For this reason, we are free to make this choice.}
$dist^2 (n,m)$ over integers $m,n$ with $n \in \{ -i,\ldots,0 , \dots, i \}$.   Let $m_0$ be the smallest  integer greater than $ 2 r^2/\tau^2 - 1/2$.  No minimizer of (\ref{dist}) can  have
$|m| > m_0$, because, otherwise, decreasing $|m|$ by one decreases $dist^2$.   
Thus, since $n \in \{ -i,\ldots, 0, \dots, i\}$,  the minimization of $dist^2(n,m)$ is a finite integer minimization problem, and both
first and second nearest neighbors can always be found.  However, we have 
non-uniqueness because  $dist^2(n,m) = dist^2(-n,-m)$, and there may be additional degeneracy
as mentioned above.

\begin{wrapfigure}{r}{0.22\textwidth}
\vspace{-5mm}
\begin{center}
 \includegraphics[width=0.2 \textwidth]{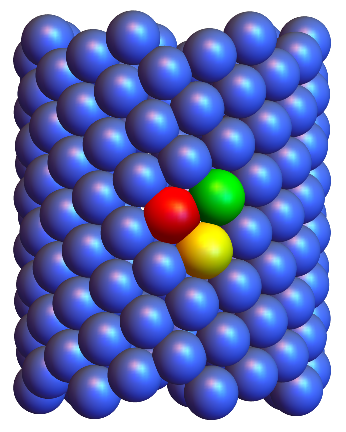}
 \caption{Illustration of nearest neighbor generators found by the algorithm (see text).  Red is mapped to
yellow by $g_1$ and red is mapped to green by $g_2$.} \label{fig:nng}
\end{center}
\vspace{-5mm}
\end{wrapfigure}

Since we have the existence of minimizers, we can suppose  a minimizer of $dist^2(m,n)$ is given by
$(n_1, m_1), \ n_1 \in \{ -i, \ldots, 0, \dots, i\}$.  Further, we can consider the auxiliary minimization problem 
\begin{align}
\min_{m,n} \{ dist^2(n,m)\colon m_1 n \neq n_1 m\},  \label{mingen}
\end{align} 
and suppose $(n_2, m_2),  \ n_2 \in \{ -i, \ldots, 0 , \dots, i\}$ is a minimizer to this problem.   Hence,  we study the group elements $g_1 = g(n_1, m_1)$ and $g_2 = g(n_2, m_2)$. Here, we call $g_1$ the \textit{nearest neighbor generator} and $g_2$ is the \textit{second nearest neighbor generator}\footnote{It is possible for $dist(n_1, m_1) = dist(n_2,m_2)$, in which case the second nearest neighbor is actually the also the nearest neighbor.}.
The meaning of the constraint $m_1 n_2 \neq m_2 n_1$
is explained in detail below, but clearly it
serves to rule out $n_2 = -n_1, m_2 = -m_1$ and other behavior such as $g_1^2 = g_2$ which
would be problematic for a concept of compatibility.

We will show that the given group $G$ is generated by the nearest neighbor generators 
$g_1$ and $g_2$.  Let  
\beq
 G' = \{g_1^p g_2^q:  (p,q) \in \Z^2 \},
 \quad {\rm (omit\ repeated\ elements)}.  \label{G'}
\eeq
Since $g_1$ and $g_2$ are both elements of the group $G$ (as well as their products), $G'$ is a subgroup of $G$.  
To show that $G' = G$, we will argue by contradiction.  The basic idea is to define a unit cell\footnote{A {\it unit cell} in this case is the direct analog of that for the translation group,
i.e., the images of the unit cell under the group cover the cylinder $ \calC$ defined just after (\ref{tangents}), and images corresponding
to distinct group elements are distinct.} based on $g_1$ and $g_2$.  Since these are nearest neighbor generators, this unit cell contains a single atom at one of the vertices.   However, in supposing that $G' \neq G$, we will argue that there must be another atom inside this unit cell.  This is the desired contradiction. 

To define the unit cell, we first note that the formula $g_1^p g_2^q$ also makes sense when $p$ and
$q$ are real numbers.  The two vectors 
\beq
\frac{d}{d \xi} g_1^{\xi}(\bfp)|_{\xi = 0}, \quad \frac{d}{d \eta} g_2^{\eta}(\bfp)|_{\eta = 0}, \label{tangents}
\eeq
define tangent vectors on the cylindrical surface $\calC = \{  \bfz + r (\cos \omega\, \bfe_1 + \sin \omega\, \bfe_2) + \zeta \bfe: 0 < \omega \le 2 \pi, \ \zeta\in \R \}$ where $\bfe_1, \bfe_2, \bfe$
are orthonormal.  These two tangents are not parallel  since by construction $m_1 n_2 \ne m_2 n_1$. We call this condition the {\it non-degeneracy condition}\footnote{If 
$m_1 n_2 = m_2 n_1$, then the two functions $g_1^{\xi}(\bfp),  g_2^{\eta}(\bfp)$ parameterize
the same curve on the cylinder.  In this case, $g_1$ and $g_2$ cannot be used to define a unit cell of the cylinder.}.  
Hence,
\beq
\calU = \{ (g_1^{\xi}(\bfp),  g_2^{\eta}(\bfp)) : 0 \le \xi < 1, \ 0 \le \eta < 1 \} \subset \cal C
\eeq
is a unit cell for $G'$ and has positive area. 

To show $G' = G$, we argue by contradiction.  We suppose that there are integers $\tilde{n}, \tilde{m}$
such that the isometry $\tilde{g} = g(\tilde{n}, \tilde{m}) \in G$ but $\tilde{g}  \notin G'$.  Let
$\xi, \eta \in \mathbb R$ satisfy
\beq
\tilde{m}  = \xi m_1  + \eta m_2  \quad \tilde{n} = \xi n_1 + \eta n_2.  \label{2x2}
\eeq
Note that (\ref{2x2})
is solvable for $(\xi, \eta) \in \mathbb R^2$ because we have assumed $m_1 n_2 \ne n_1 m_2$.
Since $\tilde{g} \notin G'$,  $(\tilde{m}, \tilde{n})$ are not both zero and at least 
one of $\xi$ and $\eta$ is not an integer.   By subtracting suitable integers from $\tilde{m}$
and $\tilde{n}$, we can assume without loss of generality that $\xi, \eta \in[-1/2, 1/2]$ and 
$\xi, \eta$ are not both zero.

We make the standing assumption on $\theta$ and $\alpha$ that $-\pi/2  \leq m_1 \theta + n_1 \alpha \leq \pi/2$ and
$-\pi/2 \leq m_2 \theta + n_2 \alpha \leq \pi/2$.  These reasonable assumptions imply 
that the unit cell $\calU$ does
not extend more than halfway around the cylinder $\calC$.  We also note that $\sin^2(\delta/2)$
is a strictly convex, even  function of $\delta$ on the interval $-\pi/2 \le \delta \le \pi/2$.  Then, using the
distance formula in (\ref{dist}), we have that
\beq
|g_1^{\xi} g_2^{\eta}(\bfp) - \bfp|^2 =4 r^2 \sin^2 \bigg( \frac{\xi (m_1 \theta + n_1 \alpha) + \eta(m_2 \theta + n_2 \alpha)}{2} \bigg)  + (\xi m_1 + \eta m_2)^2 \tau^2.
\eeq
By changing $g_1$ or $g_2$ to its inverse, if necessary (which does not change the distances
$|g_{1,2}(\bfp) - \bfp|$),  we can assume that $m_1 \ge 0, m_2 \ge 0$.

Using the evenness of $\sin^2$ (i.e., $\sin^2(\delta/2) = \sin^2(|\delta|/2)$) and its monotonicity on 
on the interval $(0, \pi/2)$, we observe that
\begin{equation}
\begin{aligned}
\sin^2 \bigg( \frac{\xi (m_1 \theta + n_1 \alpha) + \eta(m_2 \theta + n_2 \alpha)}{2} \bigg) &=
\sin^2 \bigg( \frac{| \xi (m_1 \theta + n_1 \alpha) +  \eta  (m_2 \theta + n_2 \alpha)|}{2} \bigg)  \\
& \le  \sin^2 \bigg( \frac{(1/2) |(m_1 \theta + n_1 \alpha)| + (1/2)  |(m_2 \theta + n_2 \alpha)|}{2} \bigg) \\
& \le  \frac{1}{2} \sin^2 \bigg( \frac{ m_1 \theta + n_1 \alpha}{2}\bigg) + \frac{1}{2} \sin^2\bigg( \frac{m_2 \theta + n_2 \alpha}{2}\bigg).  \label{nest}
\end{aligned}
\end{equation}
since $\xi, \eta \in [-1/2,1/2]$.  Here, the last step follows from the convexity of $\sin^2 (\delta/2)$ on $[-\pi/2, \pi/2]$.   This calculation,
together with the observation $(\xi m_1 + \eta m_2)^2 \le (1/4)(m_1 + m_2)^2$, shows that
\beqs
|g_1^{\xi} g_2^{\eta}(\bfp) - \bfp|^2 &\le& \frac{1}{2}|g_1(\bfp) - \bfp|^2 - (\tau^2/2) m_1^2 + \frac{1}{2} |g_2(\bfp) - \bfp|^2 - (\tau^2/2) m_2^2 + (\tau^2/4)(m_1 + m_2)^2 \nonumber \\
&\le& |g_2(\bfp) - \bfp|^2  - (\tau^2/4)(m_2 - m_1)^2  \nonumber \\
&\le& |g_2(\bfp) - \bfp|^2.  \label{nest1}
\eeqs

Here we have used that $g_2(\bfp)$ is a {\it second} nearest neighbor of $\bfp \in \calC$.  
Following back through the inequalities, we see that equality holds in (\ref{nest1}) only if $g_1 = g_2$ 
which is forbidden by our hypothesis $m_1 n_2 \ne m_2 n_1$. And also $g_1^{\xi} g_2^{\eta} \neq g_1$ by this 
hypothesis. Thus we reach the conclusion that $|g_1^{\xi} g_2^{\eta}(\bfp) - \bfp| < |g_2(\bfp) - \bfp|$
which contradicts that $g_2(\bfp)$ is a second nearest neighbor of $\bfp$.
Hence,  $G' = G$.

We collect these results in the form of an algorithm below:

\vspace{2mm}
\noindent {\bf Algorithm: nearest neighbor generators}. Let $G$ be given by (\ref{G}) with
group parameters $\tau \ne 0, \alpha = 2 \pi/i,\ i \in \mathbb N \setminus \{ 0\}$, 
and let $r = |\bfp - \bfz|>0$, $(\bfp - \bfz) \cdot \bfe = 0$, the unit vector $\bfe$ being
on the axis of $\bfQ_{(\cdot)}$.
Let $(n_1, m_1)$ be a minimizer of the finite-dimensional
minimization problem,
\begin{equation}
\min_{\begin{array}{l} n \in \{-i,\ldots, 0, \dots, i\} \\  m \in \mathbb Z \cap [-h,h] \end{array}}4r^2 \sin^2\! \bigg(\frac{m \theta + n \alpha}{2}\bigg)  + m^2 \tau^2 \label{alg1},
\end{equation}
and let $(n_2, m_2)$ be a minimizer of  (\ref{alg1}) subject to the constraint $m_1 n  \ne m n_1$.    Here, $h $ is the smallest integer greater than 
$ 2 r^2/\tau^2 - 1/2$.  Assume that $-\pi/2  \leq m_1 \theta + n_1 \alpha \leq \pi/2$ and
$-\pi/2 \leq m_2 \theta + n_2 \alpha \leq \pi/2$.  Then nearest neighbor generators of $G$ are given by
\beq
g_1 = g(n_1, m_1), \quad g_2 = g(n_2, m_2).
\eeq

To summarize in words, this procedure provides a nearest neighbor parameterization for \textit{any} Abelian group that generates a helical structure (i.e., those groups in (\ref{G1}) or (\ref{G3})), as long as the parameters are such that the distance between neighboring atoms is sufficiently small compared to the radius of the cylinder.

\section{Domain of powers of the nearest neighbor generators}\label{sec4}

Let  nearest neighbor generators $g_1 = g(n_1,m_1)$ and $g_2 = g(n_2,m_2)$ be given
as above so that $m_1 n_2 \ne m_2 n_1$ with $n_1, n_2 \in \{-i,\ldots,0 , \dots, i \}$.   
Note that $m_1$ and $m_2$ cannot both
be zero, so first we assume $m_1 \ne 0$.
There are always values
$(p,q) \ne (0,0)$ such that
$g_1^pg_2^q = id$.  Under our hypotheses, these are  those $(p,q) \in \Z^2$ satisfying
\beq
p n_1 + q n_2 = n i, \quad p m_1 + q m_2 = 0.  \label{int}
\eeq
Since $n_1m_2 - n_2 m_1 \ne 0$,
the solutions of this system are pairs of integers $(p,q)$ satisfying
\beq
 p = n\, i\, \Big( \frac{m_2}{n_1 m_2 - n_2 m_1}\Big), \quad
 q = n\, i\,  \Big( \frac{-m_1}{n_1 m_2 - n_2 m_1}\Big).  \label{pqsystem}
\eeq
The parameterization above of $G$ consists of powers of the two nearest neighbor generators.
A group for generating a helical structure should not have repeated elements\footnote{In general, group theory assumes an indexing with no repeated elements.}
 since we want one an only one atom at each point in the orbit. However, arbitrary powers of $g_1, g_2$ will give infinitely many pairs of powers that describe
 the same atom.   In this section, we give a general procedure for finding a suitable domain of these powers that specifies the group $G$ completely and has 
 no repeated elements.

From these solutions and the assumption $m_1 \ne 0$, 
there is a unique solution $(p^{\star}, q^{\star})$ of
$(p,q) \ne (0,0)$ of $g_1^pg_2^q = id$ that contains
the smallest positive value of $q$.  
$\calD = \Z \times \{1,\dots, q^{\star}\}$ serves as a domain for the powers $(p,q)$
with the property that there are no repeated elements. 

\begin{wrapfigure}{r}{0.22\textwidth}
	\vspace{-5mm}
	\begin{center}
		\includegraphics[width=0.2 \textwidth]{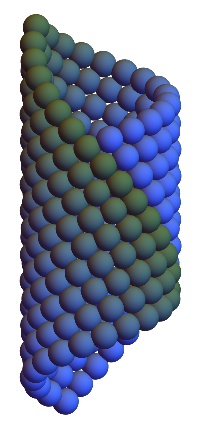}
		\caption{Illustration of the calculation of the domain $\calD$. 
			In this case $q^{\star} = 14$ and the shading
			is according to the value of $q$.} \label{fig:qstar}
	\end{center}
	\vspace{-5mm}
\end{wrapfigure} 

To see this, note first that
any $q \in \Z$ is expressible in the form $q = j q^{\star} + q'$ where $j \in \Z$ and
$q' \in \{1,\dots, q^{\star}\}$.  Thus, 
$g_1^p g_2^q = g_1^{p' + jp^{\star}}g_2^{q'+ jq^{\star}} = g_1^{p'}g_2^{q'}$,
so the powers $(p,q)$ of any group element can be assumed to lie in $\calD$.

Thus we only need to show that two distinct pairs of powers $(p,q) \in \cal D$ and $(p', q')\in \calD$ 
do not give the same group element.  To see this, assume that
$g_1^{p}g_2^{q} = g_1^{p'}g_2^{q'}$ for $q, q' \in \{1, \dots, q^{\star} \}$
and, without loss of generality, $q\ge q'$.
Then 
\beq
g_1^{p-p'}g_2^{q-q'} = id.  \label{inter}
\eeq  

By the minimality property of 
$q^{\star}$, this implies that $q = q'$, and (\ref{inter}) is
reduced to $g_1^{p-p'} = id$.  The 
latter is possible under the condition $m_1 \ne 0$
and the assumptions on the group parameters $\tau, \bfz$ if and 
only if $p= p'$.
We have shown that each point $(p,q) \in \calD$ (using nearest-neighbor
generators $g_1, g_2$) corresponds to
one and only one element of $G$.

In this argument, we made the additional hypothesis that $m_1 \ne 0$.  This is without loss of generality.   Note that 
the condition  $n_1m_2 - n_2 m_1 \ne 0$ forbids $m_1$ and $m_2$ to vanish 
simultaneously.  Thus if $m_1 = 0$, then $m_2 \neq 0$.  Consequently, we can simply replace $g_1$ with $g_2$ and $g_2$ with $g_1$ if this is the case.  
This generates the same structure.  

So far, we have deduced that $G = \{ g_1^{p} g_2^{q} \colon (p,q) \in \mathcal{D}\}$ (with $m_1 \neq 0$), and this description has no repeated elements. 
It follows from a slight modification of the proof given above that an alternative description of this group is $G = \{ g_1^{p} g_2^{q} \colon (p,q) \in \mathcal{D}_{q_0}\}$ where
$\mathcal{D}_{q_0} = q_0 + \mathcal D = \{q_0 + 1, \dots, q_0 + q^{\star} \} $ for any integer $q_0$.

We summarize these results. 

\vspace{2mm}

\noindent {\bf Algorithm: domain of powers of the generators}.  Assume the conditions on parameters
listed  above for the
derivation of nearest neighbor generators, and choose the labeling of the generators $g_1$ and $g_2$ so that $m_1 \neq 0$.  Let $\hat{q}(n)  =   -n\, i\, m_1/(n_1 m_2 - n_2 m_1)$  and define
	\beq
	q^{\star}  = \min_{\begin{array}{l} n \in \mathbb Z\setminus \{0\} \\ \hat{q} (n) > 0 \\ 
			\hat{q}(n) \in \mathbb Z \end{array}} \hat{q}(n) \label{alg2} .
	\eeq
Then for any $q_0 \in \mathbb{Z}$, $\calD_{q_0} =  \mathbb Z \times\{q_0 + 1, \dots, q_0 + q^{\star} \} $ has the property that
	$G = \{ g_1^p g_2^q : (p,q) \in \calD_{q_0} \} $, and this description has no repeated elements.
	
\vspace{2mm}

\section{Discrete vs.~continuum concepts of compatibility}\label{sec5}
Compatibility in discrete structures is fundamentally connected with Abelian groups.
It concerns the fact that a product of group elements that corresponds to a loop 
in the space of powers $\Z^n$ of the $n$ generators gives the identity. When the
group is the infinitesimal translation group operating on $\R^n$ and the action
of the group on vector fields $\bfv(\bfp), \bfp \in \R^n$ is arranged appropriately, this gives the 
usual notion of calculus, i.e., conditions under which $\bfv = \nabla \varphi$.  
Another example is given in \cite{falk_elasticity_2006} in which molecules interact according
to their positions and orientations.  

As far as we are aware, there is no theory of compatibility for interfaces between atomistic phases
having different sets of structural parameters.  This would involve the formulation
of appropriate definitions that say that one can set up a correspondence of neighboring molecules, 
such that corresponding molecules before transformation remain close after partial transformation, when
separated by a phase boundary.

On the other hand, there is a straightforward and simple notion of compatibility at interfaces for
discrete structures that is directly inherited from continuum ideas.  In the present case
it is the following. After passing to nearest neighbor parameterization as above, we
consider the formulas $\bfy_{a,b}(p,q)$ defined above. 
The functions $\bfy_a(p,q)$ and $\bfy_b(p,q)$ are defined on different
discrete domains, but they make perfect sense if they are extended to a suitable strip in $\mathbb{R}^2$.  
We wish to use continuity of these functions to impose
compatibility, and for this purpose we will extend their domains to
all values of $(p,q) \in \R \times (0, q^{\star}_{a,b}]  \subset \R^2$. (For simplicity we take $q_0=0$.)
Restrict the structure $a$ to be locally on
one side of an interface $(p,q) = (\hat{p}(s), \hat{q}(s)),\ s_1 < s < s_2$ and 
$b$ to be on the other side.
Then, under suitable smoothness assumptions, the standard continuum notion of compatibility is 

\beq
\nabla_{p,q} \bfy_a(\hat{p}(s), \hat{q}(s)) - \nabla_{p,q} \bfy_b(\hat{p}(s), \hat{q}(s))
 = \bfa(s) \otimes \bfn(s), \quad \bfn(s) = (-\hat{q}'(s), \hat{p}'(s));
\eeq 
that is, equivalently,
\beq
\Big (\partial_p \bfy_a(\hat{p}(s), \hat{q}(s)) - \partial_p \bfy_b(\hat{p}(s), \hat{q}(s))\Big)  \hat{p}'(s)  +  \Big(\partial_q \bfy_a(\hat{p}(s), \hat{q}(s)) -  \partial_q \bfy_b(\hat{p}(s), \hat{q}(s))\Big) \hat{q}'(s) = 0. \label{c}
\eeq
Note that there is a lot of freedom here.   Even with this canonical interpolation, note that 
we have complete freedom on where to place
the interface $(\hat{p}(s), \hat{q}(s))$ between the discrete positions.  In this paper
we define compatibility by using (\ref{c}).

Note that we use the same interface $(\hat{p}(s), \hat{q}(s))$ in the reference domain
for both structures in $\R^2$.  This also does not seem to be
restrictive, since we allow the group parameters as well as the $(0,0)$
positions to be assignable.  However, the interface has to respect the two
(potentially different) periods $q^{\star}_{a}$ and $q^{\star}_{b}$.
  If, say, $q^{\star}_{a} < q^{\star}_{b}$
and the interface extends into the region $q^{\star}_{a} < q < q^{\star}_{b}$ then 
$\bfy_a$ is undefined on this region: there
are no molecules from structure $a$ that can be matched with those of $b$
across this part of the interface.  
Thus we assume $0 < \hat{q}(s) \le min \{ q^{\star}_{a},q^{\star}_{b} \}$.   Moreover, we
solve rigorously the local problem: under mild hypotheses we find necessary and sufficient conditions 
that (\ref{c}) is satisfied in a sufficiently small neighborhood  $s_1 < s <  s_2$ on which 
$0 < \hat{q}(s) \le min \{ q^{\star}_{a},q^{\star}_{b} \}$.  Then we show that some of these
solutions can be extended to larger intervals.

\section{Local compatibility for nearest neighbor generators}\label{sec6}

\noindent {\bf Hypotheses on the groups.}  Let $G_a \ne G_b$ be the two helical 
groups  
\beqs
G_a &=& \left\{ \bigg( \bfQ^{a}_{m \theta_{a} + n \alpha_{a}} \bigg|\
m \tau_a \bfe_a + (\bfI -  \bfQ^{a}_{m \theta_{a} + n \alpha_{a}})\bfz_a
 \bigg): m \in \Z,\ n = 1,\dots, i_a   \right\}, \nonumber \\
G_b &=& \left\{ \bigg(  \bfQ^{b}_{m \theta_{b} + n \alpha_{b}}\, \bigg|\
m \tau_b \bfe_b +\, (\bfI -  \bfQ^{b}_{m \theta_{b} + n \alpha_{b}})\bfz_b
\bigg):  m \in \Z,\ \, n = 1,\dots, i_b  \right\},
\eeqs
where 
$\bfz_{a,b}, \bfe_{a,b} \in \mathbb{R}^3$, 
the unit vector $\bfe_{a,b}$ being on the axis of $\bfQ_{(.)}^{a,b}\in$ SO(3), and  
$\tau_{a,b} \ne 0, \theta_{a,b}\in \mathbb{R}, \alpha_{a,b} = 2 \pi / i_{a,b}, i_{a,b} \in \mathbb{N}$.
The helical structures are generated by applying these groups on $\bfp_{a,b}\in \mathbb{R}^3$ with 
$r_{a,b}=|\bfp_{a,b}-\bfz_{a,b}|>0$, $(\bfp_{a,b}-\bfz_{a,b})\cdot\bfe_{a,b}=0$.
Nearest neighbor generators given by (\ref{alg1}) and (\ref{alg2}) generate the re-parameterized groups
\beqs
G_a &=& \left\{ \bigg( \bfQ^{a}_{p \psi_{a} + q \beta_{a}} \bigg|\
(p m_1^a + q m_2^a) \tau_a \bfe_a + (\bfI -  \bfQ^{a}_{p \psi_{a} + q \beta_{a}})\bfz_a
 \bigg): p \in \Z,\ q = 1,\dots, q_a^{\star}   \right\}, \nonumber \\
G_b &=& \left\{ \bigg(  \bfQ^{b}_{p \psi_{b} + q \beta_{b}}\, \bigg|\
(p m_1^b + q m_2^b) \tau_b \bfe_b + (\bfI -  \bfQ^{b}_{p \psi_{b} + q \beta_{b}})\bfz_b
\bigg):  p \in \Z,\ \, q = 1,\dots, q_b^{\star}  \right\}\label{nngroups}
\eeqs
 with no repeated elements. 
Here $\psi_{a,b}=m_1^{a,b}\theta_{a,b}+n_1^{a,b}\alpha_{a,b}, \beta_{a,b}=m_2^{a,b} \theta_{a,b} + n_2^{a,b}\alpha_{a,b}$, and, to satisfy the algorithms for construction of the nearest neighbor generators, the integers $m_1^{a,b}$ and $m_2^{a,b}$ satisfy the condition $m_1^{a,b} n_2^{a,b} \neq m_2^{a,b} n_1^{a,b}$.  The latter is equivalent to  $m_2^{a,b}\psi_{a,b} \neq m_1^{a,b} \beta_{a,b}$ and
$-\pi/2 \le \psi_{a,b}, \beta_{a,b} \le \pi/2$.

\vspace{2mm}
\noindent {\bf Hypotheses on the interface, compatibility and orientability.} The two helical structures generated by (\ref{nngroups}) are
\beqs 
\bfy_a(p,q) &=& \bfQ^{a}_{p \psi_{a} + q \beta_{a}}\bfr_a +
(p m_1^a + q m_2^a) \tau_a \bfe_a + \bfz_a,\  p \in \Z,\ q = 1,\dots, q^{\star}_a,
\nonumber \\
\bfy_b(p,q) &=& \bfQ^{b}_{p \psi_{b} + q \beta_{b}}\bfr_b +
(p m_1^b + q m_2^b) \tau_b \bfe_b + \bfz_b,\ \ \  p \in \Z,\ q = 1,\dots, q^{\star}_b \label{pos2},
\eeqs
where $\bfr_{a,b} = \bfp_{a,b} - \bfz_{a,b}$, $\bfr_{a,b} \cdot \bfe_{a,b} = 0$.  
Following the discussion of Section \ref{sec5},  we extend the domain of $\bfy_{a,b}(p,q)$ to a suitable subset of $\mathbb{R}^2$, and we seek an arclength parameterized twice continuously differentiable  curve $(\hat{p}(s),\hat{q}(s)) \in \mathbb{R}^2$ defined on $s_1 < s < s_2$ and satisfying  $0\le \hat{q}(s)\le min \{ q_a^{\star},q_b^{\star} \}$, $\hat{p}'(s)^2+\hat{q}'(s)^2=1$.
The condition that the interface can be parameterized by arclength is without loss of generality, since the condition (\ref{c}) is linear in $\hat{p}', \hat{q}'$.
 The \textit{local compatibility condition} is then given by 
\beqs
\lefteqn{\hat{p}'(s) \Big( \bfQ^a_{\hat{p}(s)\psi_a + \hat{q}(s) \beta_a}( \psi_a \bfW^a \bfr_a + m_1^a \tau_a \bfe_a)
 -              \bfQ^b_{\hat{p}(s)\psi_b + \hat{q}(s) \beta_b}( \psi_b \bfW^b \bfr_b + m_1^b \tau_b \bfe_b) \Big)} \nonumber \\
 & &   =
-\hat{q}'(s) \Big( \bfQ^a_{\hat{p}(s)\psi_a + \hat{q}(s) \beta_a}( \beta_a \bfW^a \bfr_a + m_2^a \tau_a \bfe_a)
 -              \bfQ^b_{\hat{p}(s)\psi_b + \hat{q}(s) \beta_b}( \beta_b \bfW^b \bfr_b + m_2^b \tau_b \bfe_b) \Big),
 \label{ab1}
\eeqs
for $s \in (s_1, s_2)$ and with the hypotheses on the parameters given above.  We obtained this formula simply by substituting (\ref{pos2}) into (\ref{c}). 

Below, we identify necessary and sufficient conditions on structural parameters 
$\{ \psi_{a,b},\beta_{a,b},\tau_{a,b}\in \mathbb{R}; \bfr_{a,b}, \bfz_{a,b}, \bfe_{a,b} \in \mathbb R^3, |\bfe_{a,b}| = 1 ; m_1^{a,b}, m_2^{a,b} \in \mathbb{Z} \}$ satisfying the restrictions given above that allow for the existence of such curves.  Note that we allow the full set of parameters to vary, so that the two helical  phases can lie on different cylinders of arbitrary positive radius, with arbitrary orientation, and lying on arbitrary axes.  Further, the structural parameters defining pitch and other characteristics are subject to only very mild restrictions.   

We also note that we have not specified that the two mappings $\bfy_a$ and $\bfy_b$ map points on opposite sides of the interface on the reference domain to opposite sides of the interface on the cylinder.  That is, in a sufficiently small neighborhood $\calD = \calD^+ \cup \calD^- \subset \mathbb R \times (0, q^{\star})$ 
of a point $(\hat{p}(s_0), \hat{q}(s_0))$  divided by the interface into disjoint  regions $\calD^+$ and $\calD^-$, we may have the situation that 
compatible deformations $\bfy_a(\calD^-)$
and $\bfy_b(\calD^+)$ map to overlapping regions on the cylinder.  In this case we consider helical structures  given by
$\bfy_a(\calD^-) \cup \bfy_b(\calD^-)$ (or $\bfy_a(\calD^+) \cup \bfy_b(\calD^+)$).  This is consistent with the idea that 
nanotubes are not typically synthesized by deformations of a flat sheet of atoms, i.e, the functions (\ref{pos2}) 
do not represent actual  deformations, but just parameterizations.  The condition of compatibility is still reasonable in
this case in that a point $(p,q) \in \mathbb Z^2 \cap \calD^-$ near the interface is mapped by $\bfy_a$ and $\bfy_b$ to a nearby point
on the cylinder by compatibility, and so one can set up a 1-1 correspondence of nearby points of the two phases.

However, it is useful for the comparison with the ``rolling-up construction'' \cite{dumitrica_objective_2007} (disallowing folding), and for our analysis of 
slips and twins in Section \ref{sec8}, to distinguish the two cases.  Therefore we will say that the parameterization of a 
compatible interface is {\it orientable} if 
\beq
(\bfy_a,_p \times\, \bfy_a,_q)  \cdot  (\bfy_b,_p \times\, \bfy_b,_q )>0 \ \ {\rm on} \ (\hat{p}(s), \hat{q}(s)),  \ \ s_1 < s < s_2.  \label{orient}
\eeq 
 We introduce the vectors
\begin{align}\label{eq:fabgab}
 \mathbf{f}_{a,b} =  \left(\begin{array}{c} \psi_{a,b} \\ \beta_{a,b} \end{array}\right), \quad \mathbf{g}_{a,b} = \tau_{a,b} \left(\begin{array}{c} m_1^{a,b} \\ m_2^{a,b} \end{array}\right)
\end{align}
in order to consolidate the parameters.  In terms of $\bff_{a,b}\in [-\pi/2, \pi/2] \times [-\pi/2, \pi/2] $, $\bfg_{a,b}\in \R^2 $ and  $\mathbf{x}(s) = (\hat{p}(s), \hat{q}(s))$ the formulas (\ref{pos2}) give an explicit form of the condition (\ref{orient}) for an orientable interface:
\beq
(\bfy_a,_p \times\, \bfy_a,_q)  \cdot  (\bfy_b,_p \times\, \bfy_b,_q ) (s)= (\bff_a \cdot \bfg_a^{\perp}) (\bff_b \cdot \bfg_b^{\perp}) \big( \bfr_a  \cdot \bfQ_{(\bff_b - \bff_a) \cdot \bfx(s)} \bfr_b \big) >0,
\eeq
where we use the standard notation that $\perp$ denotes a counterclockwise rotation about $\bfe$ by $\pi/2$.

\vspace{2mm}
\noindent {\bf Simplification of the local compatibility condition.}  The compatibility condition (\ref{ab1}) can be simplified.  To do this, we first  consolidate some of the notation.  

As just above, we set $\mathbf{x}(s) = (\hat{p}(s), \hat{q}(s))$ for the arc-length parameterized curves. Thus, we have $|\mathbf{x}'(s)|=1$ for all $s \in (s_1,s_2)$, and so we define 
\begin{align}\label{eq:tangent}
\mathbf{t}(s) = \mathbf{x}'(s), \quad \mathbf{n}(s)  =  \bft(s)^{\perp} = \Big(\begin{array}{c}-\hat{q}'(s) \\  \hat{p}'(s)\end{array}\Big)
\end{align}
The tangent $\mathbf{t}(s)$ to the interface and normal $\mathbf{n}(s)$ forms an orthonormal basis at each point $\mathbf{x}(s)$, and
\begin{align}\label{eq:curvature}
\mathbf{t}'(s) = \kappa(s) \mathbf{n}(s),  \quad s_1 < s < s_2, 
\end{align}
where $\kappa(s) \in \mathbb{R}$ gives the \textit{curvature} of the interface at each $\mathbf{x}(s)$.  
The nondegeneracy conditions $m_2^{a,b}\psi_{a,b} \neq m_1^{a,b} \beta_{a,b}$ we have assumed above on group parameters are:
\beq
\bff_{a,b} \cdot \bfg_{a,b}^{\perp} \ne 0, \quad \bfg_{a,b} \cdot \bfh = 0  \ \ {\rm for\ some}  \ \ \bfh \in \Z^2\setminus \{ \mathbf{0}\},  \label{nonde}
\eeq
where $\bfg_{a,b}^{\perp} = \tau_{a,b} (-m^{a,b}_2, m^{a,b}_1)$.  We take these conditions to hold throughout.

The local compatibility equation can now be written as a total derivative,  
\begin{align}\label{eq:compat}
\Big( \mathbf{Q}^{a}_{\mathbf{x}(s) \cdot \mathbf{f}_a}\mathbf{r}_a +  (\mathbf{x}(s) \cdot \mathbf{g}_a) \mathbf{e}_a  \Big)' =  
\Big(\mathbf{Q}^{b}_{\mathbf{x}(s) \cdot \mathbf{f}_b} \mathbf{r}_b  + (\mathbf{x}(s) \cdot \mathbf{g}_b)\mathbf{e}_b \Big)' 
\end{align}
for $s \in (s_1, s_2)$.
\vspace{2mm}

We can catalogue all ways of  obtaining a locally compatible interface based on properties of the parameterized interface $\mathbf{x}(s)$ (see (\ref{eq:tangent}) and (\ref{eq:curvature})).  First note that either $\mathbf{t}(s) = \mathbf{t} = const.$ on $(s_1, s_2)$ or there is a point $s^* \in (s_1, s_2)$
where the curvature is nonzero.  In the latter case, since we are solving the local problem, we shrink the interval $(s_1, s_2),\  s_1 < s^* < s_2$ so that $\kappa(s) \ne 0$ on $(s_1, s_2)$.
Thus, for the {\it locally compatible interface} there are two cases to consider:
\begin{enumerate}
\item $\mathbf{t}(s) = \mathbf{t} = const.$ for all $s \in (s_1,s_2)$.
\item The curvature $\kappa(s)  \neq 0$ for all $s \in ({s}_1, {s}_2)$. 
\end{enumerate}

In the remainder of this section, we develop a complete characterization of local compatibility.  We state this characterization in the form of theorems for each of the two cases above.  
These results are then summarized and discussed in the next section.

\subsection{Implications of vanishing interface curvature}
\begin{lemma}\label{curveLemma}
Assume the hypotheses on the helical groups $G_a \neq G_b$ and on the interface.  The curvature $\kappa(s) = 0$ for all $s \in (s_1, s_2)$ if and only if 
the cylinders are parallel, $\mathbf{e}_a \times \mathbf{e}_b = 0$.  
\end{lemma}
\begin{proof}
($\Leftarrow$) Without loss of generality, we can assume $\mathbf{e}_a = \mathbf{e}_b = \mathbf{e}$ since the case $\mathbf{e}_a = - \mathbf{e}_b$ has an identical structure (after replacing $\mathbf{f}_a$ with $-\mathbf{f}_a$ and $\mathbf{g}_a$ with $-\mathbf{g}_a$).  After explicit differentiation and some rearranging of terms,  (\ref{eq:compat}) becomes
\begin{align}\label{eq:compatLemma}
(\mathbf{t}(s) \cdot \mathbf{f}_a) \mathbf{Q}_{\mathbf{x}(s) \cdot \mathbf{f}_a} \mathbf{W}\mathbf{r}_a  - (\mathbf{t}(s) \cdot \mathbf{f}_b) 
\mathbf{Q}_{\mathbf{x}(s) \cdot \mathbf{f}_b} \mathbf{W}\mathbf{r}_b = \Big((\mathbf{g}_b - \mathbf{g}_a) \cdot \mathbf{t}(s)\Big) \mathbf{e}
\end{align}
for $s \in (s_1, s_2)$.  Here, $\mathbf{Q}_{(\cdot)} = \mathbf{Q}^{a}_{(\cdot)} = \mathbf{Q}^b_{(\cdot)}
$ and $\mathbf{W} = \mathbf{W}^{a} = \mathbf{W}^b$ since $\mathbf{e}_{a,b} = \mathbf{e}$  (recall (\ref{getQ})).  Both $\mathbf{Q}_{\mathbf{x}(s) \cdot \mathbf{f}_b} \mathbf{W}\mathbf{r}_b$ and  $\mathbf{Q}_{\mathbf{x}(s) \cdot \mathbf{f}_a} \mathbf{W}\mathbf{r}_a$ are perpendicular to $\mathbf{e} = \mathbf{e}_{a,b}$.  Thus, equality holds if and only if both sides vanish.  For the right-hand side, $(\mathbf{g}_b - \mathbf{g}_a) \cdot \mathbf{t}(s) = 0$.  By differentiating this identity, we have that $\kappa(s) (\mathbf{g}_b - \mathbf{g}_a) \cdot \mathbf{n}(s) = 0$.  Since $\mathbf{n}(s)$ is perpendicular to $\mathbf{t}(s)$, we conclude 
$\kappa(s) = 0$ for all $s \in (s_1, s_2)$, or  $\mathbf{g}_a = \mathbf{g}_b$.  If the latter condition \textit{does not} hold, then the proof is complete.  Hence, we assume $\mathbf{g}_a = \mathbf{g}_b$.  

Substituting $\mathbf{g}_a = \mathbf{g}_b$ back into (\ref{eq:compatLemma}), the left-hand side vanishes.  
There are then two distinct cases to consider:
\begin{enumerate}
\item[(a)]  The vectors $\mathbf{Q}_{\mathbf{x}(s) \cdot \mathbf{f}_b} \mathbf{W}\mathbf{r}_b$ and $\mathbf{Q}_{\mathbf{x}(s) \cdot \mathbf{f}_a} \mathbf{W}\mathbf{r}_a$  are parallel for all $s \in (s_1, s_2)$.  
\item[(b)]  These vectors are linearly independent on some interval\footnote{If they are linearly independent at a single point, then, by continuity, they must be linearly independent on some interval.} $(\tilde{s}_1, \tilde{s}_2) \subset (s_1, s_2)$.
\end{enumerate}

We suppose (a). Then, we can write   $\mathbf{Q}_{\mathbf{x}(s) \cdot \mathbf{f}_a} \mathbf{W}\mathbf{r}_a = \pm  (|\mathbf{r}_a|/|\mathbf{r}_b|) \mathbf{Q}_{\mathbf{x}(s) \cdot \mathbf{f}_b} \mathbf{W}\mathbf{r}_b$, where the $\pm$ is fixed for all $s$ (given the smoothness hypothesis).  This implies $\mathbf{W} \mathbf{r}_a = \pm  (|\mathbf{r}_a|/|\mathbf{r}_b|) \mathbf{Q}_{\mathbf{x}(s) \cdot (\mathbf{f}_b - \mathbf{f}_a)} \mathbf{W} \mathbf{r}_b$ since $\mathbf{e}_{a,b} = \mathbf{e}$.  By differentiating this quantity, we deduce  the identity $\mathbf{t}(s) \cdot (\mathbf{f}_b - \mathbf{f}_a) = 0$.  By differentiating this, we deduce that either $\kappa(s) = 0$ on for all $s$ or $\mathbf{f}_b = \mathbf{f}_a$.  We assume the latter, as the former is desired.  In substituting this back into (\ref{eq:compatLemma}) (using that $\mathbf{g}_a = \mathbf{g}_b$ and $\mathbf{f}_a = \mathbf{f}_b \neq 0$), we conclude $\mathbf{r}_a = \mathbf{r}_b$ since $\mathbf{Q}_{\mathbf{x}(s) \cdot (\mathbf{f}_b - \mathbf{f}_a)} = \mathbf{I}$.   But then, $\mathbf{g}_a = \mathbf{g}_b$, $\mathbf{f}_a = \mathbf{f}_b$ and $\mathbf{r}_a = \mathbf{r}_b$, violating the hypothesis $G_a \ne G_b$. 

We suppose (b).  Then, $\mathbf{Q}_{\mathbf{x}(s) \cdot \mathbf{f}_b} \mathbf{W}\mathbf{r}_b$ and $\mathbf{Q}_{\mathbf{x}(s) \cdot \mathbf{f}_a}\mathbf{W}\mathbf{r}_a$ are linearly independent for all $s \in (\tilde{s}_1, \tilde{s}_2)$.  Consequently, the left-hand side of (\ref{eq:compatLemma}) vanishes on this interval if and only if $\mathbf{t}(s) \cdot \mathbf{f}_b = \mathbf{t}(s) \cdot \mathbf{f}_a = 0$ for all $s \in (\tilde{s}_1, \tilde{s}_2)$.  By differentiating these identities (as we did above), we conclude that either $\kappa(s) = 0$ for $s \in (\tilde{s}_1, \tilde{s}_2)$ or $\mathbf{f}_a = \mathbf{f}_b = 0$.  The latter contradicts the hypotheses on the group parameters stated at the start of this section.  So the former must be true in this case.  

In summary, we have shown that $(\mathbf{e}_a \times \mathbf{e}_b) =0$ implies $\kappa(s) = 0$  for all $s_1< s< s_2$, in both cases (a) and (b). 

($\Rightarrow$)  Suppose, for the sake of a contradiction, that $\kappa(s) = 0$ for all $s \in (s_1, s_2)$ but $(\mathbf{e}_a \times \mathbf{e}_b) \neq 0$.  By explicitly differentiating the compatibility condition in (\ref{eq:compat}), we obtain
\begin{align}\label{eq:1stLemma}
(\mathbf{t} \cdot \mathbf{f}_b) \mathbf{Q}^b_{\mathbf{x}(s) \cdot \mathbf{f}_b} \mathbf{W}^b \mathbf{r}_b - (\mathbf{t} \cdot \mathbf{f}_a) \mathbf{Q}^a_{\mathbf{x}(s) \cdot \mathbf{f}_a}\mathbf{W}^a \mathbf{r}_a = (\mathbf{g}_a \cdot \mathbf{t}) \mathbf{e}_a -  (\mathbf{g}_b \cdot \mathbf{t}) \mathbf{e}_b
\end{align}
for all $s \in (s_1, s_2)$.  Notice that the tangent $\mathbf{t}(s) = \mathbf{t} = const$ in this case.  Thus, by differentiating twice more, we obtain two additional equations
\begin{equation}
\begin{aligned}\label{eq:2ndLemma}
&(\mathbf{t} \cdot \mathbf{f}_b)^2 \mathbf{Q}^b_{\mathbf{x}(s) \cdot \mathbf{f}_b} \mathbf{r}_b = (\mathbf{t} \cdot \mathbf{f}_a)^2 \mathbf{Q}^a_{\mathbf{x}(s) \cdot \mathbf{f}_a} \mathbf{r}_a, \\
&(\mathbf{t} \cdot \mathbf{f}_b)^3 \mathbf{Q}^b_{\mathbf{x}(s) \cdot \mathbf{f}_b}\mathbf{W}^b \mathbf{r}_b = (\mathbf{t} \cdot \mathbf{f}_a)^3 \mathbf{Q}^a_{\mathbf{x}(s) \cdot \mathbf{f}_a} \mathbf{W}^a \mathbf{r}_a,
\end{aligned}
\end{equation}
which must hold for all $s \in (s_1, s_2)$.  We dot both of these equations with $\mathbf{e}_b$ so that the left-hand sides vanish.  Then, we must have $\mathbf{t} \cdot \mathbf{f}_a  = 0$ or $\mathbf{e}_b \cdot \mathbf{Q}^a_{\mathbf{x}(s) \cdot \mathbf{f}_a} \mathbf{W}^a \mathbf{r}_a = \mathbf{e}_b \cdot \mathbf{Q}^a_{\mathbf{x}(s) \cdot \mathbf{f}_a} \mathbf{r}_a = 0$.  However, the latter implies that $\mathbf{e}_b$ is parallel to $\mathbf{e}_a$ since the set $\{  \mathbf{Q}^a_{\mathbf{x}(s) \cdot \mathbf{f}_a} \mathbf{r}_a,  \mathbf{Q}^a_{\mathbf{x}(s) \cdot \mathbf{f}_a}\mathbf{W}^a \mathbf{r}_a, \mathbf{e}_a\}$ forms an orthogonal basis of $\mathbb{R}^3$.  But $\mathbf{e}_a \times \mathbf{e}_b \neq 0$ by hypothesis.  So we conclude $\mathbf{f}_a \cdot \mathbf{t} = 0$.  Now, we instead dot the equations in (\ref{eq:2ndLemma}) with $\mathbf{e}_a$ so that the right-hand sides vanish.  By a similar argument, we conclude $\mathbf{f}_b \cdot \mathbf{t} = 0$.   Substituting $\mathbf{f}_{a,b} \cdot \mathbf{t} = 0$ back into (\ref{eq:1stLemma}), we see that the left-hand side vanishes.  It, therefore, follows that $\mathbf{g}_a \cdot \mathbf{t} = \mathbf{g}_b \cdot \mathbf{t} = 0$ since $\mathbf{e}_a \times \mathbf{e}_b \neq 0$.  In summary, we have deduced that $\mathbf{f}_{a,b}$ and $\mathbf{g}_{a,b}$ are all parallel for this case.  But this means that $\mathbf{f}_a \cdot \mathbf{g}_a^{\perp} = 0$, which violates the non-degeneracy condition (\ref{nonde}).  This is the desired contradiction, proving that, if $\kappa(s) = 0$ for all $s \in (s_1, s_2)$, then $(\mathbf{e}_a \times \mathbf{e}_b) =0$.
\end{proof}

\subsection{Classification of all local solutions}
\label{sect6.2}

Locally compatible interfaces with zero curvature or nonzero curvature can now be determined.  Recall
the notation (\ref{eq:fabgab})-(\ref{nonde}).

\subsubsection{Vertical, horizontal and helical interfaces }

\begin{theorem}[Interfaces with zero curvature]\label{VHHThrm} 
Assume the hypotheses on the helical groups $G_a \neq G_b$, and assume without loss of generality that $\bfe_a \cdot \bfe_b > 0$.  
Assume that the curvature $\kappa(s) = 0$ on $(s_1, s_2)$.   Each locally compatible interface is contained in one of the following cases:
\begin{enumerate}
\item[(i)] (Vertical interfaces).   $\mathbf{e}_a = \mathbf{e}_b$, $\mathbf{t} \cdot \mathbf{f}_a = \mathbf{t} \cdot \mathbf{f}_b = 0$, and  $\mathbf{t} \cdot \mathbf{g}_a =  \mathbf{t} \cdot \mathbf{g}_b \neq 0$; 
\item[(ii)] (Horizontal interfaces). $\mathbf{e}_a = \mathbf{e}_b$, $\mathbf{t} \cdot \mathbf{f}_a = \mathbf{t} \cdot \mathbf{f}_b \neq 0$, $\mathbf{g}_a \cdot \mathbf{t} = \mathbf{g}_b \cdot \mathbf{t} = 0$, and
\begin{align}\label{eq:rotTheorem}
\mathbf{r}_a = \mathbf{Q}^b_{ \mathbf{x}(s_1) \cdot (\mathbf{f}_b - \mathbf{f}_a)} \mathbf{r}_b;
\end{align}
\item[(iii)] (Helical interfaces). The same as the horizontal interface except that $\mathbf{g}_a \cdot \mathbf{t} = \mathbf{g}_b \cdot \mathbf{t} \neq 0$.  
\end{enumerate}
The interface  is given by  $\mathbf{x}(s) = (s-s_1) \mathbf{t} + \mathbf{c}$ for some $\mathbf{c} \in \mathbb{R}^2$ and $\mathbf{t} \in \mathbb{S}^1$.   
In all cases $\bfz_b$ and $\bfz_a$ are restricted  by matching the formulas (\ref{pos2}) at one point on the interface, e.g., $\bfx(t_1)$.
\end{theorem}
\begin{proof}  Note that (\ref{nonde}) implies that $\bff_{a,b} \ne 0$.  Since $\bft(s) = \bft = const$, we have $\mathbf{e}_a \times \mathbf{e}_b \neq 0$ by Lemma \ref{curveLemma} and so $\mathbf{e}_a = \mathbf{e}_b = \mathbf{e}$.  Thus, $\mathbf{Q}^{a}_{(\cdot)} = \mathbf{Q}^b_{(\cdot)} = 
\mathbf{Q}_{(\cdot)}$ and $\mathbf{W}^a = \mathbf{W}^b = \mathbf{W}$.  Hence, by explicitly differentiating the compatibility equation  (\ref{eq:compat}) and pre-multiplying this equation by $\mathbf{Q}_{-\mathbf{x}(s) \cdot \mathbf{f}_a}$, we obtain a condition equivalent to local compatibility in this case: 
\begin{align}\label{eq:compatConst}
(\mathbf{t} \cdot \mathbf{f}_b) \mathbf{Q}_{\mathbf{x}(s) \cdot (\mathbf{f}_b - \mathbf{f}_a)} \mathbf{W}\mathbf{r}_b - (\mathbf{t} \cdot \mathbf{f}_a)\mathbf{W} \mathbf{r}_a = (\mathbf{g}_a \cdot \mathbf{t} -  \mathbf{g}_b \cdot \mathbf{t}) \mathbf{e}
\end{align}
for $s \in (s_1, s_2)$.  Here, $\mathbf{W} \mathbf{r}_a$ and $ \mathbf{Q}_{\mathbf{x}(s) \cdot (\mathbf{f}_b - \mathbf{f}_a)} \mathbf{W}\mathbf{r}_b$ are both perpendicular to $\mathbf{e}$.  Thus, by dotting this quantity with $\mathbf{e}$, we see that $\mathbf{g}_a \cdot \mathbf{t}$ must equal $\mathbf{g}_b \cdot \mathbf{t}$. This condition is necessary in all cases as stated in the theorem.  

Substituting this back into (\ref{eq:compatConst}), the right-hand side vanishes.  Then, by differentiating this equation, we obtain the necessary condition
\begin{equation}
\begin{aligned}\label{eq:cc1}
&(\mathbf{t} \cdot \mathbf{f}_b) (\mathbf{t} \cdot  (\mathbf{f}_b - \mathbf{f}_a)) \mathbf{Q}_{\mathbf{x}(s) \cdot (\mathbf{f}_b - \mathbf{f}_a)} \mathbf{r}_b = 0
\end{aligned}
\end{equation}
for $s \in (s_1, s_2)$.  By assumption  $\mathbf{r}_b$ cannot be zero, so $\mathbf{Q}_{\mathbf{x}(s) \cdot (\mathbf{f}_b - \mathbf{f}_a)} \mathbf{r}_b$ is not zero.  Thus, there are only two possibilities for local compatibility in this case:
\begin{enumerate}
\item[(a)] Either, $\mathbf{t} \cdot \mathbf{f}_b = 0$;
\item[(b)] Or, $\mathbf{t} \cdot \mathbf{f}_b \ne 0$ and $\mathbf{t} \cdot  (\mathbf{f}_b - \mathbf{f}_a) = 0$.
\end{enumerate}

We suppose (a).  Then, in substituting $\mathbf{t} \cdot \mathbf{f}_b = 0$ back into the first in (\ref{eq:compatConst}), we have a locally compatible helical structure if and only if $(\mathbf{t} \cdot \mathbf{f}_a)\mathbf{W} \mathbf{r}_a=0$; that is, if and only if $\mathbf{t} \cdot \mathbf{f}_a = 0$.  In combining all the identities, we obtain Case (i) in the theorem.  The inequality $\mathbf{g}_{a,b} \cdot \mathbf{t} \neq 0$ 
follows from the nondegeneracy conditions (\ref{nonde}) assumed on the group parameters.  Thus,  hypothesis (a) gives Case (i) of the theorem.

Now, we suppose (b) above.   Note that since $\mathbf{x}'(s) = \mathbf{t} = const$, the curve  is given by $\mathbf{x}(s)  = (s -s_1) \mathbf{t} + \mathbf{c}$ for some $\mathbf{c} \in \mathbb{R}^2$ (as stated in the theorem).  Making use of this fact, we observe that $\mathbf{Q}_{\mathbf{x}(s) \cdot (\mathbf{f}_b - \mathbf{f}_a) } = \mathbf{Q}_{ \mathbf{c} \cdot (\mathbf{f}_b - \mathbf{f}_a) } = const.$ since $\mathbf{t} \cdot  (\mathbf{f}_b - \mathbf{f}_a) = 0$.    Substituting this into (\ref{eq:compatConst}), we have a locally compatible interface if and only if 
\begin{align} \label{hi}
(\mathbf{t} \cdot \mathbf{f}_b) \mathbf{Q}_{\mathbf{c} \cdot (\mathbf{f}_b - \mathbf{f}_a)}\mathbf{W}\mathbf{r}_b =  (\mathbf{t} \cdot \mathbf{f}_a)\mathbf{W} \mathbf{r}_a.
\end{align}
Noting that $\bfQ_{(\cdot)}$ and $\bfW$ commute, so we can remove $\bfW$ from (\ref{hi}), and we can cancel the nonzero terms $\mathbf t \cdot \mathbf f_b = \mathbf t \cdot \mathbf f_a$.
Thus, we have necessarily that $|\mathbf r_b| = |\mathbf r_a|$ and $\mathbf{Q}_{ \mathbf{x}(s_1) \cdot (\mathbf{f}_b - \mathbf{f}_a) } \mathbf r_b = \mathbf r_a$.  This condition is also
sufficient for (\ref{hi}).  Therefore, necessary and sufficient conditions for local compatibility under hypotheses  (b) are given in Cases (ii) and (iii) of the theorem.
\end{proof}

\subsubsection{Elliptic interfaces} \label{elliptical}
\begin{theorem}[Interfaces with nonzero curvature] \label{EllipticThrm}
Assume the hypotheses on the helical groups $G_a \neq G_b$, assume the hypothesis on the interface, and assume that the curvature $\kappa(s) \ne 0$ on $(s_1, s_2)$.  
Introduce an
orthonormal basis $\bfe_1 = \bfe_b \times \bfe_a/|\bfe_b \times \bfe_a|$, $\bfe_2 = (\bfe_b - \bfe_a)/|\bfe_b - \bfe_a|$ and $\bfe_3 = \bfe_1 \times \bfe_2$.  
Define $0 \le \theta_{a,b} < 2 \pi$ and $\rho_{a,b}>0$ by $\rho_a \bfQ^a_{\theta_a}\bfe_1 = {\bfr}_a $ and $\rho_b \bfQ^b_{\theta_b} \bfe_1 = {\bfr}_b$.  Note that  $\bfe_a = -\sin \xi \bfe_2 + \cos \xi \bfe_3, \ \bfe_b = \sin \xi \bfe_2 + \cos \xi \bfe_3$ for a suitable $0 < \xi < \pi$, $\xi \ne \pi/2$. 
Each locally compatible interface is contained in one of the following cases:
	\begin{enumerate}
		\item[(i)] (Type 1 Elliptic interfaces).   $\mathbf{f}_a = \mathbf{f}_b$,  $\mathbf{g}_a =  -\mathbf{g}_b$, $\rho_a = \rho_b$, $\theta_a = \theta_b$.  In non-arclength parameterization the interface is given by
\beq 
\tilde{\bfx}(t) =  t \bfu -  (\rho_a \tan \xi \sin (t + \theta_a)  + c )\bfv,  \label{i1}
\eeq
where $\bfu \cdot \bfg_b = \bfv  \cdot \bff_b = 0, \bfu \cdot \bff_b = \bfv \cdot \bfg_b =1$, $c \in \mathbb{R}$ and $t \in (t_1, t_2)$ for the interval chosen below. 
		\item[(ii)] (Type 2 Elliptic interfaces).  $\mathbf{f}_a = -\mathbf{f}_b$,  $\mathbf{g}_a =  \mathbf{g}_b$, $\rho_a = \rho_b$, $\theta_a = -\theta_b$.  In non-arclength parameterization the interface is given by
\beq
\tilde{\bfx}(t) =  t \bfu -  (\rho_a \cot \xi \sin (-t + \theta_a) + c  ) \bfv  \label{i2}
\eeq
with $\bfu, \bfv, c$ as above and $t \in (t_1, t_2)$ for the interval below.
	\end{enumerate}
\end{theorem}
The formulas for the interface can be converted to arclength parameterization by the standard method of writing $\bfx(s)=\tilde{\bfx}(t(s))$ where $t(s)$ is the inverse of
$s(t) = \int_{t_1}^t |\tilde{\bfx}'(\tilde{t})|d\tilde{t} +s_1$ and $t_2$ is such that $s(t_2) = s_2$.

\begin{proof}  The introduction of  the basis $(\bfe_1, \bfe_2, \bfe_3)$ and $\xi$ are justified by $\bfe_ a \times \bfe_b \ne 0$ (see Lemma \ref{curveLemma})
and  $\bfe_1 \cdot \bfe_{a,b} =0$.
Consider the integrated form of (\ref{eq:compat}),
\beq
\bfQ^a_{\bfx(s) \cdot \bff_a + \theta_a} \Big( \rho_a \bfe_1+ (\bfx(s) \cdot \bfg_a)\bfe_a \Big) = \bfQ^b_{\bfx(s) \cdot \bff_b + \theta_b}  \Big( \rho_b \bfe_1+ (\bfx(s) \cdot \bfg_b)\bfe_b \Big)
+ \bfc.  \label{ecc1}
\eeq
The proof consists of identifying certain components of (\ref{ecc1}).  Dotting (\ref{ecc1}) by $\bfe_1$, $\bfe_2$ and $\bfe_3$, respectively,  yields the equations
\begin{equation}
\begin{aligned}
\rho_a  \cos (\bfx(s) \cdot \bff_a + \theta_a) &= \rho_b  \cos (\bfx(s) \cdot \bff_b + \theta_b) + c_1 , \\
\rho_a \sin (\bfx(s) \cdot \bff_a + \theta_a) -  \tan \xi (\bfx(s) \cdot \bfg_a) &= \rho_b  \sin (\bfx(s) \cdot \bff_b + \theta_b) + \tan \xi (\bfx(s) \cdot \bfg_b) + c_2/ \cos \xi,  \\
\rho_a \sin (\bfx(s) \cdot \bff_a + \theta_a) +  \cot \xi (\bfx(s) \cdot \bfg_a) &= -\rho_b  \sin (\bfx(s) \cdot \bff_b + \theta_b) +\cot \xi (\bfx(s) \cdot \bfg_b) + c_3/ \sin \xi.
\label{ecc3}
\end{aligned}
\end{equation}

We first show that the nonzero vectors $\bff_a$ and $\bff_b$ are parallel.  Suppose, for the sake of a contradiction, they are not.  Then they are linearly independent, and
so there exist linearly independent reciprocal vectors $\bff^a, \bff^b$ satisfying $\bff_a \cdot \bff^a = 1, \bff_b \cdot \bff^b = 1, \bff_a \cdot \bff^b =
\bff_b \cdot \bff^a = 0$.  Therefore, we express the tangent to the interface in the reciprocal basis: 
\begin{equation}
\begin{aligned}\label{eq:reciprocal}
\mathbf{t}(s) = \eta_a(s) \bff^a + \eta_b(s) \bff^b, 
\end{aligned}
\end{equation} 
where the parameters $\eta_{a,b}(s)$ are continuously differentiable due to the hypothesis on the interface.  
We differentiate (\ref{ecc3})$_2$ and (\ref{ecc3})$_3$ with respect to $s$, eliminate  $\cos (\bfx(s) \cdot \bff_a + \theta_a)$
in both cases using (\ref{ecc3})$_1$, and add and subtract the resulting equations.  This gives
\beq
\big(2 \rho_b \cos (\bfx(s) \cdot \bff_b + \theta_b) \bff_a + \bfk_1\big) \cdot \mathbf{t}(s) = 0, \quad \big(2 \rho_b \cos \big(\bfx(s) \cdot \bff_b + \theta_b) \bff_b + \bfk_2\big) \cdot \mathbf{t}(s) = 0 ,  \label{eec4}
\eeq
where $\bfk_1$ and $\bfk_2 $ are explicit functions of $\xi, c_1, \bff_a, \bfg_a, \bfg_b$.  To further simplify,
we insert (\ref{eq:reciprocal}) into (\ref{eec4}), multiply
(\ref{eec4})$_1$ by $\eta_b(s)$ and (\ref{eec4})$_2$ by $\eta_a(s)$, and add to get
\beq
\eta_a(s)^2 (\bfk_2 \cdot \bff^a)  + \eta_a(s) \eta_b(s) (\bfk_1 \cdot \bff^a + \bfk_2 \cdot \bff^b) + \eta_b(s)^2 (\bfk_1 \cdot \bff^b) = 0. \label{eq:qeta}
\eeq
Since the curvature $\kappa(s)$ is non-zero, there exists an $(\tilde{s}_1,\tilde{s}_2) \subset (s_1, s_2)$ such that $\eta_a(s)$ and $\eta_a'(s) \neq 0$ on this sub-interval.  We divide through by $\eta_a(s)^2 \neq 0$ in (\ref{eq:qeta}) to obtain a quadratic equation in $\lambda(s) = \eta_b(s)/\eta_a(s)$ on this reduced interval.   If $\lambda'(s) \neq 0$ on $(\tilde{s}_1,\tilde{s}_2)$, then it immediately follows that the coefficients of this quadratic equation must vanish, i.e., 
\begin{equation}
\begin{aligned}\label{eq:degenQuad}
\bfk_1 \cdot \bff^b = \bfk_2 \cdot \bff^a = \bfk_1 \cdot \bff^a + \bfk_2 \cdot \bff^b = 0.
\end{aligned}
\end{equation}
The derivative is indeed non-vanishing: We notice that $\mathbf{t}(s)/\eta_a(s) = \mathbf{f}^a + \lambda(s) \mathbf{f}^b$ by definition, and consequently, differentiating this quantity yields the desired result since both $\kappa(s)$ and $\eta_a'(s)$ are non-vanishing on this interval.  Hence, (\ref{eq:degenQuad}) is a necessary condition on the parameters.  These equations are then solved by expressing $\bfk_1, \bfk_2$ in the basis $\bff_a, \bff_b$, yielding
\beq
\bfk_1 = \delta \bff_a, \quad \bfk_2 = -\delta \bff_b, \quad \text{ for some } \delta \in \R. \label{ecc5}
\eeq
We insert (\ref{ecc5}) into (\ref{eec4}) to obtain
\beq
(2 \rho_b \cos (\mathbf{x}(s) \cdot \mathbf{f}_b + \theta_b)  + \delta)  \eta_a(s) = 0, \quad (-2 \rho_b \cos (\mathbf{x}(s) \cdot \mathbf{f}_b + \theta_b)  - \delta)  \eta_b(s) = 0.  \label{eec6}
\eeq
Since $\eta_a(s) \neq 0$ on $(\tilde{s}_1, \tilde{s}_2)$, we observe that $\cos(\mathbf{x}(s)\cdot \mathbf{f}_b + \theta_b ) = const$ on this interval.  By the smoothness hypothesis of the interface, it follows that $\mathbf{x}(s) \cdot \mathbf{f}_b =  const$ on this interval.  By differentiation and the parameterization for $\mathbf{t}(s)$ in (\ref{eq:reciprocal}), we conclude $\eta_b(s) =0$ on this interval.  But this means that $\eta_a(s) = 1/|\mathbf{f}^a|$ on this interval since the tangent is a unit vector.  This contradicts the fact that $\eta_a'(s) \neq 0$ on $(\tilde{s}_1, \tilde{s}_2)$.  Therefore, $\bff_a$ and $\bff_b$ are in fact parallel.

Let $\bff_b = \lambda \bff_a$ for some $\lambda \ne 0$ (recall (\ref{nonde})). We show that $\lambda = \pm 1$.  We work in a sufficiently
small neighborhood of $s^* \in (s_1, s_2)$ where $\mathbf{t}(s^*) \cdot \bff_a \ne 0$.   Under our smoothness assumptions we can differentiate 
(\ref{ecc3})$_1$ as many times as we like near $s^*$ as long as we cancel the $\mathbf{t}(s) \cdot \bff_a$ after each differentiation.  Comparing the
first and third derivative of (\ref{ecc3})$_1$ we get that $\lambda^2=1$, so $\bff_b = \pm \bff_a$.  We treat these two cases separately.

Suppose $\bff_b = \bff_a$.  Again we work near $s = s^*$.  Comparing the
first and second derivative of (\ref{ecc3})$_1$, we get that  $\rho_a (\cos (\bfx(s) \cdot \bff_a + \theta_a), \sin (\bfx(s) \cdot \bff_a + \theta_a))= 
\rho_b (\cos (\bfx(s) \cdot \bff_a + \theta_b), \sin (\bfx(s) \cdot \bff_a + \theta_b))$.  Since $\rho_{a,b}>0$ and $0 \le \theta_{a,b} < 2 \pi$, this shows that $\rho_a = \rho_b$ and $\theta_a = \theta_b$. The derivative of (\ref{ecc3})$_2$ with respect to $s$ near $s^*$ implies immediately that $\bfg_a = -\bfg_b$. At this point (\ref{ecc3})$_{1, 2}$ are satisfied with $c _1 =c_2= 0$ and (\ref{ecc3})$_{3}$ becomes a condition that determines $\bfx(s)$.  Consider an arbitrary regular  parameterization $\tilde{\mathbf{x}} \colon (t_1, t_2) \rightarrow \mathbb{R}^2$ given by $\tilde{\mathbf{x}}(t) = \zeta_1(t) \mathbf{u} + \zeta_2(t) \mathbf{v}$, where $(\bfu, \bfv)$ are reciprocal vectors to
the linearly independent vectors $\bff_b, \bfg_b$ as defined by  $\bfu \cdot \bfg_b = \bfv  \cdot \bff_b = 0, \bfu \cdot \bff_b = \bfv \cdot \bfg_b =1$.  We substitute $\mathbf{x}(s(t)) = \tilde{\mathbf{x}}(t)$ into (\ref{ecc3})$_3$ (for $s(t) = \int_{t_1}^{t} |\tilde{\mathbf{x}}'(\tilde{t})| d \tilde{t} + s_1$ with $t_2$ such that $s(t_2) = s_2$) to derive the necessary and sufficient conditions on the arc-length parameterized curve $\mathbf{x} \colon (s_1, s_2) \rightarrow \mathbb{R}^2$.  We find $\zeta_2(t) = \rho_b \tan \xi \sin( \zeta_1(t) + \theta_b) + c$ for any $c \in \mathbb{R}$  solves (\ref{ecc3})$_3$.  This shows that the interface is the graph of a function in the direction $\mathbf{f}_b$, and so without loss of generality, we can set $\zeta_1(t) = t$ and parameterize by arclength to obtain a generic expression for the interface curve $\mathbf{x}(s)$ satisfying (\ref{ecc3})$_3$.  This is given by (\ref{i1}) in the theorem.

Suppose on the other hand $\bff_b = -\bff_a$.   Again comparing the
first and second derivative of (\ref{ecc3})$_1$, we now get that  $\rho_a = \rho_b$ and $\theta_a = -\theta_b$.  In this case
the derivative of (\ref{ecc3})$_3$ with respect to $s$ near $s^*$ implies that $\bfg_a = \bfg_b$, and (\ref{ecc3})$_2$ gives
the formula (\ref{i2}) for the interface.
\end{proof}

Examples of vertical, horizontal, helical and elliptic interfaces consistent with Theorems \ref{VHHThrm} and \ref{EllipticThrm} are given in Figure \ref{fig:interfaces}.  All four types of interface can be extended to be global solutions in typical cases.  By substituting the formulas for the interface
given in Theorem \ref{elliptical} into the formula (\ref{pos2}), one can prove that the elliptic interfaces are indeed ellipses in the
the helical configuration.  Elliptic interfaces are typically not orientable in the sense of (\ref{orient}): this explains the appearance of the
overlapping reference domains in Figure \ref{fig:interfaces}d, which have been displaced from each other to be easily visible.  Figure
\ref{fig:interfaces}d illustrates motion of the elliptical interface, but typically vertical and helical interfaces cannot be moved.  
We discuss further the mobility of interfaces in Section \ref{sec9}.

\begin{figure}[ht!]
	\centering
	\includegraphics[width=\textwidth]{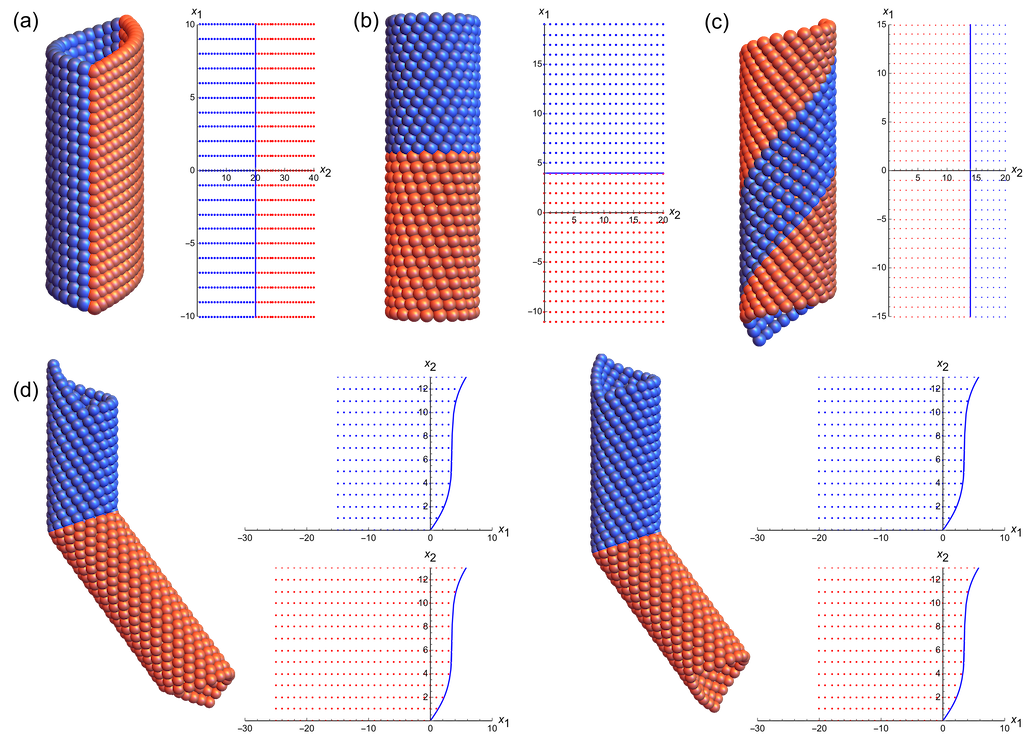}
	\caption{Examples of vertical, horizontal, helical, and elliptical interfaces ((a)-(d)) between phase $a$ (blue) and phase $b$ (red) shown in the deformed configuration (left) and reference domain (right; interfaces in blue).  The local solutions are extended to global loops, where possible. They correspond to the choices $|\bfr_a|=|\bfr_b|=1$ (all cases) and 
(a)  $\bfx(s) =  (s, 20)$,  $\bff_a = (0, 0.25), \bff_b = (0, 0.15),  \bfg_a = (0.28, -0.08),  \bfg_b = (0.28,0.084)$
(b)  $\bfx(s) = (4,s)$, \ $\bff_a = (0.12,\pi/10), \bff_b = (-0.4, \pi/10), \bfg_a = (0.262, 0),  \bfg_b = (0.3,0)$, 
(c)  $\bfx(s) = (s,14)$, $\bff_a = (0.2,0.22), \bff_b = (0.2, 0.33), \bfg_a = (0.27, -0.18),  \bfg_b = (0.27,0.081)$, 
(d)  $\tilde{\bfx}(t) =  t(3.5/\pi, 7/\pi) + \tan(\pi/10) \sin(t)  (10/3,-5/3)$ (non-arclength parameterization, see Theorem \ref{EllipticThrm}), $\bff_a = \bff_b = (2\pi/35, 4 \pi/35), \bfg_a = -\bfg_b = (-0.24,0.12)$. 
}
	\label{fig:interfaces}
\end{figure}

\section{Local compatibility using near neighbor generators}
\label{sec7}

Nearest neighbor generators provide  convenient descriptors for a helical structure, as neighbors in the $\mathbb{Z}^2$ lattice are neighboring points in the helical structure.  These generators have other key properties as shown Sections \ref{sec3} and \ref{sec4}: (i) they can be explicitly obtained for any discrete Abelian helical group under mild assumptions, and (ii) they have a suitable reference configuration. 
By examining analogs in helical structures of the concepts of slip and twinning in crystals,  we have noticed that the study of compatibility using  a fixed set of nearest neighbor 
generators for each phase is too restrictive: some of the
excluded cases are interesting.  These cases include  examples of additional compatible interfaces (analogs of twins) obtained by switching to a 
different choice of nearest neighbor generators for one of the phases. 
Therefore, in this section we are led to consider a certain precise notion of 
{\it near neighbor generators} with the properties (i) and (ii), and to  study the resulting
compatible interfaces. The concept of compatibility used here with near neighbor generators is the same as the one used above.

\subsection{Lattice invariant transformations}
\label{alt}

We first recall the basic invariance of the $\mathbb{Z}^2$ lattice \cite{pz_book_02}.  The set of invertible linear transformations mapping $\mathbb{Z}^2$ to $\mathbb{Z}^2$ is
\begin{align}\label{eq:Lattice1}
GL(\mathbb{Z}^2) = \Big\{ \boldsymbol{\mu} \in \mathbb{R}^{2\times2} \colon \mu_{ij} \in \mathbb{Z}, \; i =1,2,\; j = 1,2, \ \  \det \boldsymbol{\mu} \in \{ \pm 1\} \Big\}. 
\end{align}
Consider the nearest neighbor parameterizations of the two helical phases given by (\ref{pos2}), and bring out the dependence of these formulas on the group parameters by writing 
these positions as $\bfy_a(p,q) = \bfy(p,q; \bff_a, \bfg_a, \mathbf{r}_a, \bfz_a)$ and $\bfy_b(p,q) =  \bfy(p,q; \bff_b, \bfg_b, \mathbf{r}_b, \bfz_b)$.
Each element $\boldsymbol{\mu} \in GL(\mathbb{Z}^2) $ gives an alternative parameterization of these same two given helical phases by replacing $(p,q) = \boldsymbol{\mu}(\tilde{p},\tilde{q})$,
with $(\tilde{p},\tilde{q})$ in the domains $\tilde{\calD}_{a,b} = \boldsymbol{\mu}^{-1} (\Z \times \{1, \dots, q^{\star}_{a,b}\})$, respectively.  
As can be seen from the formulas (\ref{pos2}), the matrix $\boldsymbol{\mu}$
can be moved onto the group parameters.  The positions
\begin{align}\label{eq:transpq}
\bfy(\tilde{p},\tilde{q}; \boldsymbol{\mu}^{T}\bff_a, \boldsymbol{\mu}^{T}\bfg_a, \mathbf{r}_a, \bfz_a), \ \  (\tilde{p},\tilde{q}) \in \tilde{\calD}_a, \quad \bfy(\tilde{p},\tilde{q}; \boldsymbol{\mu}^{T}\bff_b, \boldsymbol{\mu}^{T}\bfg_b, \mathbf{r}_b, \bfz_b), \ \  (\tilde{p},\tilde{q}) \in \tilde{\calD}_b,
\end{align}
are therefore the same two sets of atomic positions as given by $\bfy_a(p,q), \bfy_b(p,q)$.  If the two phases are compatible across an interface $(p(s), q(s))$, then the positions 
(\ref{eq:transpq}) are compatible across the interface $(\tilde{p}(s), \tilde{q}(s)) = \boldsymbol{\mu}^{-1}(p(s), q(s))$.

As can be seen from Section \ref{sec3}, nearest neighbor generators are defined using a particular helical structure, i.e., a particular choice of $(\bff, \bfg)$. 
Thus, it may happen that the formulas for nearest neighbor generators (given just after (\ref{mingen})) evaluated at the new group 
parameters  $(\boldsymbol{\mu}^{T}\bff, \boldsymbol{\mu}^{T}\bfg)$  are not nearest neighbor generators.   Also, it can happen that formulas for non-nearest neighbor 
generators evaluated at particular choices of  $\bff$ and $\bfg$ give nearest neighbor generators when evaluated at $(\boldsymbol{\mu}^{T}\bff, \boldsymbol{\mu}^{T}\bfg)$.

In summary, given certain formulas for generators $g_1, g_2$ of the group $G$ evaluated at $(\bff, \bfg)$, then those formulas evaluated at $(\boldsymbol{\mu}^{T}\bff, \boldsymbol{\mu}^{T}\bfg)$
give exactly the same helical structure.   However, if $g_1, g_2$ evaluated at $(\bff, \bfg)$ are nearest neighbor generators, then  $g_1, g_2$ evaluated at 
$(\boldsymbol{\mu}^{T}\bff, \boldsymbol{\mu}^{T}\bfg)$ are not generally nearest neighbor generators, and vice versa.  If the {\it same} $\boldsymbol{\mu} \in GL(\mathbb{Z}^2) $ is applied
to compatible phases $a$ and $b$, then they remain compatible, the interface is unchanged in the helical configuration,  but the description of the interface in  terms of $g_1, g_2$ evaluated at 
$(\boldsymbol{\mu}^{T}\bff, \boldsymbol{\mu}^{T}\bfg)$ changes to $(\tilde{p}(s), \tilde{q}(s)) = \boldsymbol{\mu}^{-1}(p(s), q(s))$.

We can also transform group parameters using {\it different} elements of $GL(\Z^2) $ for the two lattices.  By the discussion above, we can without loss of generality
leave one lattice unchanged, since applying a $\boldsymbol{\mu} \in GL(\mathbb{Z}^2)$ to both sets of generators is a equivalent to a change of reference configuration.   Let us fix the nearest neighbor generators of phase $a$ so that the group parameter are $(\mathbf{f}_a, \mathbf{g}_a)$ and apply $\boldsymbol{\mu} \in GL(\Z^2)$ to the the nearest neighbor generators of phase $b$ so that the group parameters are $(\boldsymbol{\mu}^T\mathbf{f}_b, \boldsymbol{\mu}^T \mathbf{g}_b)$.  Again, if we use
the appropriate domains of integers, both structures are exactly the same.  However, the meaning of the compatibility conditions changes, because nearby pairs of integers do not in general give nearby points in the helical structure of phase $b$.  

This is not a problem\textemdash when thinking in terms of phase transformations\textemdash as long as nearby points on the helical structure prior to transformation remain {\it reasonably close} after transformation.  If we ignore the fact that $\boldsymbol{\mu} \in GL(\mathbb{Z}^2)$ and think of $\boldsymbol{\mu}$ and a general $2 \times 2$ matrix, the formulas (\ref{mingen})$_{\rm ff}$  for generators depend smoothly on  $\boldsymbol{\mu}$.  Thus, a simple measure of ``reasonable closeness'' is  $|\boldsymbol{\mu} - \bfI |$.  

With these physical considerations in mind, we can  generalize the local compatibility condition, Theorem \ref{VHHThrm},  for vertical, horizontal and helical interfaces to include other choices of generators beyond those for nearest neighbor generators $\{ \mathbf{f}_a, \mathbf{g}_a\}$ and $\{ \mathbf{f}_b, \mathbf{g}_b\}$.  Referring to Theorem \ref{VHHThrm}, the condition is 
\begin{equation}
\begin{aligned}\label{eq:CompatEq}
\Big((\mathbf{f}_b, \mathbf{g}_b)^T \boldsymbol{\mu} - (\mathbf{f}_a, \mathbf{g}_a)^T \Big)\mathbf{t} = 0
\end{aligned}
\end{equation}
from some unit tangent $\mathbf{t}$ and appropriate inequalities that define the subcases vertical, horizontal and helical. (We have simply replaced $\{ \mathbf{f}_b, \mathbf{g}_b\}$ in Theorem \ref{VHHThrm} with $\{ \boldsymbol{\mu}^T \mathbf{f}_b, \boldsymbol{\mu}^T \mathbf{g}_b\}$.)  Note that the nondegeneracy condition (\ref{nonde}) is satisfied for $\{ \boldsymbol{\mu}^T \mathbf{f}_b, \boldsymbol{\mu}^T \mathbf{g}_b\}$ if and only if it is satisfied for $\{  \mathbf{f}_b,  \mathbf{g}_b\}$.

We  focus on vertical, horizontal and helical interfaces here, because no additional locally compatible elliptic interfaces are obtained if we include a lattice invariant transformation $\boldsymbol{\mu} \in GL(\mathbb{Z}^2)$ of phase $b$.  
Also, as discussed at the beginning of Section \ref{sec6}, recall that opposite sides of the reference interface need not be mapped to opposite
sides of the deformed interface.

Following the remark above about ``near closeness'' we define {\it near neighbor generators} as those associated to $\boldsymbol{\mu} \in GL(\mathbb{Z}^2)$ of the form
\begin{equation}
\begin{aligned}
\mathcal{N}(GL(\mathbb{Z}^2)) = \left\{ \left(\begin{array}{cc} \sigma_1 & \sigma_2 \\ \sigma_3 & \sigma_4 \end{array}\right) \colon \sigma_{1,2,3,4} \in \{ \pm 1, 0\}, \;\; \sigma_1 \sigma_4 - \sigma_2 \sigma_3 \in \{ \pm 1\} \right\}.
\end{aligned}
\end{equation}
When the two phases $a$ and $b$ are the same these represent slip by one lattice spacing, or twinning.
An easy enumeration shows that 
\beq
 \mathcal{N}(GL(\mathbb{Z}^2)) = \mathcal{N}^{(+)}(GL(\mathbb{Z}^2)) \cup \mathcal{N}^{(-)}(GL(\mathbb{Z}^2)),   \label{N}
 \eeq
 where
\begin{equation}
\begin{aligned}
\mathcal{N}^{(+)}(GL(\mathbb{Z}^2)) &= \Big\{ \pm \left(\begin{array}{cc} 1 & 0 \\ 0 & 1 \end{array}\right), \pm \left(\begin{array}{cc} 0 & -1 \\ 1&  0  \end{array}\right), \pm \left(\begin{array}{cc} 1 & -1 \\ 1&  0  \end{array}\right), \pm \left(\begin{array}{cc} 1 & 0 \\ -1 &  1  \end{array}\right) ,  \pm \left(\begin{array}{cc} 0 & 1 \\ -1 &  1  \end{array}\right),\\
&\qquad   \pm \left(\begin{array}{cc} 1 & 1 \\ -1&  0  \end{array}\right), \pm \left(\begin{array}{cc} 1 & -1 \\ 0 &  1  \end{array}\right),  \pm \left(\begin{array}{cc} 1 & 1 \\ 0 &  1  \end{array}\right),  \pm \left(\begin{array}{cc} 1 & 0 \\ 1 &  1  \end{array}\right) , \pm \left(\begin{array}{cc} 0 & -1 \\ 1 &  1  \end{array}\right) \Big\};  \\
\mathcal{N}^{(-)}(GL(\mathbb{Z}^2)) &= \left(\begin{array}{cc} 0 & 1\\1 & 0 \end{array}\right)\mathcal{N}^{(+)}(GL(\mathbb{Z}^2)),
\end{aligned}  \label{Npm}
\end{equation}
and the superscript $(\pm)$ indicates the sign of the determinant.

\section{Slip and twinning in helical structures}
\label{sec8}

In this section, we consider the two phases $a$ and $b$ to be the same, in the sense that they
are related by an orthogonal transformation and translation.  In this case, compatible deformations are analogous to slip or twinning.  
Our main question is whether we can have compatible vertical, horizontal or helical interfaces, between
two copies of the same phase that are oriented differently.  

\subsection{Local compatibility of helical structures in the same phase}
\label{same}

In this section we define precisely what it means that phase $b$ is
the same phase as phase $a$.  Let phase $a$ be given with nearest neighbor group parameterization $\{\bff, \bfg\}$, where we drop the subscript ``$a$'' for simplicity.  
The positions of phase $a$ are $\bfy(p,q; \bff, \bfg,\mathbf{r}, \bfz)$, where  $(p,q) \in \calD =  \Z \times \{1, \dots, q^{\star}\}$.  In view of Theorem \ref{VHHThrm} , we have extended the definition of $\bfy$  to
$\mathbf{x}=(p,q) \in \calD^c = \R \times (0, q^{\star})$.

Guided by the basic invariance of quantum mechanics\textemdash orthogonal transformations with determinant $\pm 1$ and translations\textemdash we consider
a second copy of phase $a$ given by
\beq
\hat{\bfQ} \bfy(\mathbf{x}; \boldsymbol{\mu}^T\bff, \boldsymbol{\mu}^T\bfg,\mathbf{r}, \bfz) + \hat{\bfc}, \quad \hat{\bfQ} \in {\rm O(3)},\ \  \hat{\bfc} \in \R^3,  \ \ \mathbf{x}  = (p,q) \in \boldsymbol{\mu}^{-1} \calD^c,
\eeq
where we have allowed for a change to near neighbor generators by introducing $\boldsymbol{\mu} \in \mathcal{N}^{(\sigma)}(GL(\mathbb{Z}^2))$.  

We seek a locally compatible vertical,
horizontal and helical interfaces between
$\bfy(\mathbf{x}; \bff, \bfg, \mathbf{r}, \bfz)$ and the copy $\hat{\bfQ} \bfy(\mathbf{x}; \boldsymbol{\mu}^T\bff, \boldsymbol{\mu}^T\bfg,\mathbf{r}, \bfz) + \hat{\bfc}$  at an interface $\mathbf{x}(s) \in \calD^c \cap \boldsymbol{\mu}^{-1}\calD^c$ for $s \in (s_1, s_2)$.  
 In these cases, the two phases have a common axis, so that necessarily $\hat{\bfQ} \bfe = \pm \bfe$.   The four families of $\hat{\bfQ} \in$ O(3)
 satisfying $\hat{\bfQ} \bfe = \pm \bfe$ are
 \beqs
 \det \hat{\bfQ} &=& +1\ \  \Longrightarrow   \left\{ \begin{array}{l}   \hat{\bfQ} = -\bfI + 2 \bfe^{\perp} \otimes  \bfe^{\perp}, \ \ \bfe^{\perp} \cdot \bfe =0, \ \ 
 |\bfe^{\perp}| = 1,  \ \  \hat{\bfQ} \bfQ = \bfQ^T \hat{\bfQ},\   or \\  \hat{\bfQ} = \bfR_{\bfe} , \ \ \bfR_{\bfe} \bfe = \bfe, \ \  \bfR_{\bfe} \in {\rm SO(3)}, \ \  \hat{\bfQ} \bfQ = \bfQ \hat{\bfQ}, \end{array} \right. \nonumber \\
   \det \hat{\bfQ} &=& -1 \ \   \Longrightarrow   \left\{ \begin{array}{l}   \hat{\bfQ} = \bfI - 2 \bfe^{\perp} \otimes  \bfe^{\perp}, \ \ \bfe^{\perp} \cdot \bfe =0, \ \  |\bfe^{\perp}| = 1, \ \  \hat{\bfQ} \bfQ = \bfQ^T \hat{\bfQ}, \  or \\
   \hat{\bfQ} = -\bfR_{\bfe} , \ \ \bfR_{\bfe} \bfe = \bfe, \ \  \bfR_{\bfe} \in {\rm SO(3)},  \ \  \hat{\bfQ} \bfQ = \bfQ \hat{\bfQ}. \end{array} \right. \label{cases}
 \eeqs
 The second copy of phase $a$ can then be expressed in the form
\beq
\hat{\bfQ} \bfy(\bfx ; \boldsymbol{\mu}^T\bff, \boldsymbol{\mu}^T\bfg, \mathbf{r}, \bfz) + \hat{\bfc} = \bfy(\mathbf{x}; \pm \boldsymbol{\mu}^T\bff, (\pm)\boldsymbol{\mu}^T\bfg, \hat{\mathbf{Q}} \mathbf{r}, \hat{\bfQ}\bfz + \hat{\mathbf{c}}), \quad  \bfx \in \boldsymbol{\mu}^{-1} \calD^c.  \label{copy}
\eeq
Here, the $\pm$ arises because $\hat{\bfQ} \bfQ_{\theta} =   \bfQ_{\pm \theta} \hat{\bfQ}$ for some choice of $\pm$  in all cases of (\ref{cases});
the other (independent) choice $(\pm)$ arises from $\hat{\bfQ} \bfe = (\pm) \bfe$.
We identify phase $b$ with the copy of phase $a$ described in (\ref{copy}) and characterize solutions to the  compatibility conditions (\ref{eq:CompatEq}) corresponding to vertical, horizontal and helical interfaces for near neighbor generators.  We make two observations that simplify the analysis below.  
\begin{enumerate}
\item We  drop $(\pm)$ in (\ref{copy}).  This is justified as long as we analyze all near neighbor generators, or an appropriate subset
that is invariant under multiplication by $\pm 1$, (see (\ref{N}), (\ref{Npm})).
\item For horizontal and helical interfaces the condition  $\mathbf{r}_a = \mathbf{Q}^b_{ \mathbf{x}(s_1) \cdot (\mathbf{f}_b - \mathbf{f}_a)} \mathbf{r}_b$
of Theorem \ref{VHHThrm} can be satisfied for all cases of (\ref{cases}).  That is, we satisfy $\bfr_b = \hat{\bfQ} \bfr_a = \hat{\bfQ} \bfr$ by choosing $\bfR_{\bfe}$ or $\bfe^{\perp}$ in (\ref{cases})  appropriately, i.e., choose $\hat{\bfQ} =  \mathbf{Q}^b_{ -(\mathbf{x}(s_1) \cdot (\mathbf{f}_b - \mathbf{f}_a))}$.  We assume that this is done in all cases below where we discuss helical or horizontal interfaces. For these cases the condition (\ref{orient})
of orientability becomes
\beq
(\bfy_a,_p \times\, \bfy_a,_q)  \cdot  (\bfy_b,_p \times\, \bfy_b,_q ) (s)=  \pm (\det \boldsymbol{\mu}) (\bff \cdot \bfg^{\perp})^2
\big( \bfr \cdot \bfQ_{((\pm \boldsymbol{\mu}^T - \bfI) \cdot \bft )(s-s_1)} \bfr \big)>0. \label{orient1}
\eeq	
\end{enumerate}
In Sections \ref{subsec:slip} and \ref{subsec:twin} below we treat separately the two cases in which the $b$ phase generators are given by $\{\boldsymbol{\mu}^T \mathbf{f}, \boldsymbol{\mu}^T \mathbf{g}\}$ and $\{-\boldsymbol{\mu}^T \mathbf{f}, \boldsymbol{\mu}^T \mathbf{g}\}$.   It will emerge below that these cases correspond to ``slips" and ``twins", respectively.

\subsection{Examples: Slips}
\label{subsec:slip}
In this section we choose $+$ of $\pm$ that occurs in Section \ref{same}.  Thus, the group parameters are $\{ \mathbf{f}_a, \mathbf{g}_a\} = \{ \mathbf{f}, \mathbf{g}\}$, which are required to satisfy the
nondegeneracy conditions  (\ref{nonde}), and $\{ \mathbf{f}_b, \mathbf{g}_b\} = \{ \boldsymbol{\mu}^T \mathbf{f}, \boldsymbol{\mu}^T \mathbf{g} \}$.
Thus, to obtain a locally compatible interface in the sense of Theorem  \ref{VHHThrm}, 
we need to find a $\boldsymbol{\mu} \in \mathcal{N}(GL(\mathbb{Z}^2))$ and $\mathbf{t} \in \mathbb{S}^1$ such that  $(\mathbf{f},\mathbf{g})^T (\boldsymbol{\mu} - \mathbf{I}) \mathbf{t} = 0$.  Since $(\mathbf{f}, \mathbf{g})^T$ is invertible by  (\ref{nonde}), the latter
\begin{figure*}[ht!]
	\centering
	\includegraphics[width=0.8\textwidth]{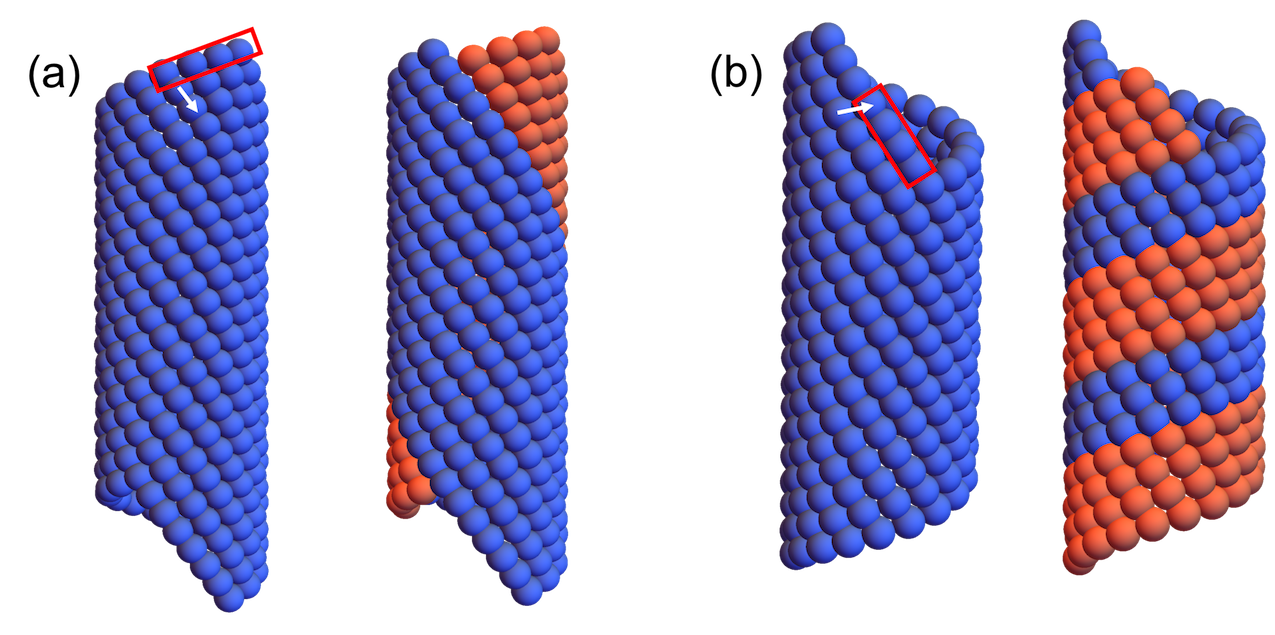}
	\caption{Examples of slips.  (a). Along the nearest neighbor generator:  $\mathbf{t}= \mathbf{e}_1$, $\boldsymbol{\mu} = (\mathbf{e}_1,\mathbf{e}_1 + \mathbf{e}_2)$, $\mathbf{f} = (2\pi/35)(1,2)$, $\mathbf{g} =  0.12(-2,1)$. (b). Along the second nearest neighbor generator: $\mathbf{t} = \mathbf{e}_2$, $\boldsymbol{\mu} = (\mathbf{e}_1+\mathbf{e}_2, \mathbf{e}_2)$, $\mathbf{f} = (2\pi/35)(1,2)$, $\mathbf{g} =  0.12(-2,1)$.  The arrows indicate slip vectors (see (\ref{eq:slipVector})). }
	\label{fig:slips}
\end{figure*}
is equivalent to $(\boldsymbol{\mu} - \mathbf{I})\mathbf{t} = 0$.  
Notice that the existence of a slip is independent of the generators of the helical structure $\{\mathbf{f}, \mathbf{g} \}$.  In other words, if a slip exists for one helical structure, then it is universal in the sense that an analogous slip exists for all the others.  However, as in crystals, the loading (e.g., analogs of the Schmid stress) and atomic forces may favor some slips over others.

We first discuss orientable interfaces.  Under the conditions assumed here, the formula (\ref{orient1}) for orientability simplifies to 
$\det \boldsymbol{\mu} > 0$, and so we seek solutions for $\boldsymbol{\mu} \in \mathcal{N}^{(+)}(GL(\mathbb{Z}^2))$, excluding
cases when there is no slip (i.e., $\boldsymbol{\mu} = \mathbf{I}$).  
We assume $0 \leq (\mathbf{t} \cdot \mathbf{e}_2) \leq 1$ without loss of generality.  
There are four cases in total:
\begin{enumerate}
\item[(i).]  Slip along the nearest neighbor generator, i.e., 
\begin{itemize}
\item $\mathbf{t} =\left(\begin{array}{c} 1 \\ 0 \end{array}\right), \quad \boldsymbol{\mu} = \left\{ \left(\begin{array}{cc} 1 &  1 \\ 0 & 1\end{array}\right),  \left(\begin{array}{cc} 1 & - 1 \\ 0 & 1\end{array}\right) \right\}$;
\end{itemize}
\item[(ii).]  Slip along the second nearest neighbor generator, i.e, 
\begin{itemize}
\item  $\mathbf{t} =\left(\begin{array}{c} 0 \\ 1 \end{array}\right),\quad \boldsymbol{\mu} = \left\{\left(\begin{array}{cc} 1 &  0 \\ 1 &  1\end{array}\right), \left(\begin{array}{cc} 1 &  0 \\ -1 &  1\end{array}\right)\right\}$.
\end{itemize}
\end{enumerate}
The case (i)  represent slip along closest-packed lines, i.e., lines of nearest neighbor atoms, and (ii) selects slip along next nearest
neighbor atoms.  Thus, the condition of orientability nicely selects the
closest-packed directions among near neighbor generators, which are expected to be favored by helical structures, as are 
the close-packed $\{111\}$ family of planes in FCC crystals\footnote{The elementary reasoning in the two cases is the same: these are lines
of atoms with least corrugation.}.  

An elegant formula can be given for the  tangent to the interface on the helical structure $\mathbf{t}_{\mathbf{y}}$ that encodes all the parameters\footnote{We choose $\mathbf{x}(s_1) = \mathbf{0}$ in this formula without loss of generality.}:
\begin{equation}
\begin{aligned}\label{eq:slipVector}
\mathbf{t}_{\mathbf{y}} = \frac{\nabla \mathbf{y}_a(\mathbf{0}) \mathbf{t}}{|\nabla \mathbf{y}_a(\mathbf{0}) \mathbf{t}|} = \frac{(\mathbf{f} \cdot \mathbf{t}) \mathbf{W}\mathbf{r} + (\mathbf{g} \cdot \mathbf{t}) \mathbf{e}}{\sqrt{|\mathbf{r}|^2 (\mathbf{f} \cdot \mathbf{t})^2 + (\mathbf{g} \cdot \mathbf{t})^2}}. 
\end{aligned}
\end{equation}
The slip vectors corresponding to the examples in Figure \ref{fig:slips} are given by:
\begin{table}[!ht]
	\centering
	\begin{threeparttable}
		\label{tab:slips_physical}
		\begin{tabular}{cccc}
			\toprule
			Slip Vector &	 Slip (a)		  &  Slip (b) 	\\
			\midrule 	
		$\mathbf{t}_{\mathbf{y}}$ &(0, 0.598974, -0.800769)  & (0, 0.948429, 0.316989)   \\
			\bottomrule
		\end{tabular}
		
	\end{threeparttable}
\end{table}

Some non-orientable cases are also interesting, particularly for special loadings or special parameters.  For example, in Figure \ref{fig:slips} a
slip in the direction of any of the six atoms surrounding any atom might be considered reasonable, depending on atomic forces and
loading.  Also, we have six nearest neighbors.  These {\it close packed helical structures} have non-orientable solutions of $(\boldsymbol{\mu} - \mathbf{I})\mathbf{t} = 0$, e.g.,  $\left(\begin{array}{rr} 0 & -1\\ -1 & 0 \end{array}\right) \in \mathcal{N}^{(-)}(GL(\mathbb{Z}^2)$ with $\bft = (1,-1)$.

\subsection{Examples: Twins.}
\label{subsec:twin}
In this section we choose $-$ of $\pm$ that occurs in Section \ref{same}. 
Thus, the group parameters are $\{ \mathbf{f}_a, \mathbf{g}_a\} = \{ \mathbf{f}, \mathbf{g}\}$ and $\{ \mathbf{f}_b, \mathbf{g}_b\} = \{ -\mathbf{f}, \mathbf{g}\}$, and they satisfy the non-degeneracy condition in (\ref{nonde}).   Thus, to obtain an {\it orientable locally compatible twin}, we require $\sigma = -\sign\big( (\mathbf{f} \cdot \mathbf{g}^{\perp})^2\big) = -1$, and we seek a $\boldsymbol{\mu} \in \mathcal{N}^{(-)}(GL(\mathbb{Z}^2))$ and $\mathbf{t} \in \mathbb{S}^1$ such that  $((-\mathbf{f},\mathbf{g})^T \boldsymbol{\mu} - (\mathbf{f},\mathbf{g})^T) \mathbf{t} = 0$ and
$(-\mathbf{f},\mathbf{g})^T \boldsymbol{\mu} \ne  (\mathbf{f},\mathbf{g})^T$ (distinct phases).  As indicated
by the terminology, these solutions represent helical analogs of twinning.   The solutions depend fundamentally on the given generators (i.e., akin to $\lambda_2 = 1$ for martensitic phase transformations \cite{chen_study_2013,gu_17}).  This dependence can be catalogued based on the type of interface.
\begin{lemma}[Twinning Lemma]\label{TwinningLemma}
Let $\{ \mathbf{f}_{a,b}, \mathbf{g}_{a,b}\}$ as above, subject to (\ref{nonde}) and $\sigma = -1$.
A helical structure with these parameters can form an orientable locally compatible twin, i.e.,
$((-\mathbf{f},\mathbf{g})^T \boldsymbol{\mu} - (\mathbf{f},\mathbf{g})^T) \mathbf{t} = 0$ for $\mathbf{t} \in \mathbb{S}^1$ and $\boldsymbol{\mu} \in \mathcal{N}^{(-)}(GL(\mathbb{Z}^2))$, with
$(-\mathbf{f},\mathbf{g})^T \boldsymbol{\mu} \ne  (\mathbf{f},\mathbf{g})^T$,   if and only if $ \boldsymbol{\mu}^T (-\mathbf{f},\mathbf{g}) \ne  (\mathbf{f},\mathbf{g})$ and
\begin{enumerate}
\item[(i).] (Vertical Twin)  $\boldsymbol{\mu} \in \mathcal{N}^{(-)}(GL(\mathbb{Z}^2))$ and $\mathbf{t} \in \mathbb{S}^1$ satisfy $\boldsymbol{\mu} \mathbf{t} = \mathbf{t}$ and $\mathbf{t}^{\perp} \parallel  \mathbf{f}$. 
\item[(ii).] (Horizontal Twin) $\boldsymbol{\mu} \in \mathcal{N}^{(-)}(GL(\mathbb{Z}^2))$ and $\mathbf{t} \in \mathbb{S}^1$ satisfy $\boldsymbol{\mu} \mathbf{t} = -\mathbf{t}$ and $\mathbf{t}^{\perp} \parallel  \mathbf{g}$. 
\item[(iii).] (Helical Twin)   $\boldsymbol{\mu} \in \mathcal{N}^{(-)}(GL(\mathbb{Z}^2))$ and $\mathbf{t} \in \mathbb{S}^1$  are such that
\begin{equation}
\begin{aligned}
\text{$\mathbf{t}$ is not an eigenvector of $\boldsymbol{\mu}$,} \quad \bff \parallel ((\boldsymbol{\mu} + \bfI)\bft)^{\perp}, \ \ {\rm and} \ \ 
\bfg \parallel ((\boldsymbol{\mu} - \bfI)\bft)^{\perp}.
 \end{aligned}
 \end{equation}
\end{enumerate}
\end{lemma}
We postpone the proof to the end of this section and instead classify the solutions given by Lemma \ref{TwinningLemma}.    
We assume $0 \leq (\mathbf{t} \cdot \mathbf{e}_2) \leq 1$ without loss of generality. Where $\bff$ or $\bfg$ are not assigned they are 
subject only to the hypotheses of Lemma \ref{TwinningLemma}. There are several cases:
\begin{enumerate}
\item[(i).] (Vertical Twins) 
\begin{itemize}
\item  $\mathbf{t} = \left(\begin{array}{c} 1 \\ 0 \end{array}\right), \quad \boldsymbol{\mu} \in \Big\{ \left(\begin{array}{cc} 1 & 1 \\ 0 & - 1 \end{array}\right),   \left(\begin{array}{cc} 1 & -1 \\ 0 & -1 \end{array}\right),  \left(\begin{array}{cc} 1 & 0 \\ 0 & -1 \end{array}\right) \Big\}, \quad \mathbf{f} \parallel \left(\begin{array}{c}  0 \\ 1 \end{array} \right);$
\item  $\mathbf{t} = \left(\begin{array}{c} 0 \\ 1 \end{array}\right), \quad \boldsymbol{\mu} \in \Big\{ \left(\begin{array}{cc} -1 & 0 \\ 1 & 1 \end{array}\right),   \left(\begin{array}{cc} -1 & 0 \\ -1 & 1 \end{array}\right),  \left(\begin{array}{cc} -1 & 0 \\ 0 & 1 \end{array}\right) \Big\}, \quad \mathbf{f} \parallel \left(\begin{array}{c}  1 \\ 0 \end{array} \right);$
\item  $\mathbf{t} = \frac{1}{\sqrt{2}}\left(\begin{array}{c} \pm1 \\ 1 \end{array}\right), \quad \boldsymbol{\mu} = \left(\begin{array}{cc} 0 & \pm1  \\ \pm1 & 0 \end{array}\right), \quad \mathbf{f} \parallel \left(\begin{array}{c}  1 \\ \mp 1 \end{array} \right), \quad \text{respectively};$
\item  $\mathbf{t} = \frac{1}{\sqrt{5}}\left(\begin{array}{c} \pm 2  \\ 1 \end{array}\right), \quad \boldsymbol{\mu} = \left(\begin{array}{cc} 1 & 0  \\ \pm1 & -1 \end{array}\right), \quad \mathbf{f} \parallel \left(\begin{array}{c}  \mp1 \\ 2 \end{array} \right), \quad \text{respectively};$
\item  $\mathbf{t} = \frac{1}{\sqrt{5}}\left(\begin{array}{c} \pm1  \\  2 \end{array}\right), \quad \boldsymbol{\mu} = \left(\begin{array}{cc} -1 & \pm1  \\ 0 & 1 \end{array}\right), \quad \mathbf{f} \parallel \left(\begin{array}{c}  2 \\ \mp1 \end{array} \right), \quad \text{respectively}.$
\end{itemize}
\item[(ii).] (Horizontal Twins)  For each of the above, replace $\boldsymbol{\mu}$ by $-\boldsymbol{\mu}$ and switch $\mathbf{f}$ and $\mathbf{g}$.
\item[(iii).] (Helical Twins)  Assume $m,n \in \mathbb{Z}$ with $n \geq 0$, $m^2 + n^2 \neq 0$. 
\begin{itemize}
\item $\mathbf{t} = \frac{1}{\sqrt{n^2+m^2}} \left(\begin{array}{c} m \\ n \end{array}\right), \quad \boldsymbol{\mu} = \left(\begin{array}{cc} 1 & 1 \\1 & 0  \end{array}\right), \quad \mathbf{f} \parallel \left(\begin{array}{c} - m - n \\ 2m+n\end{array}\right), \quad  \mathbf{g} \parallel \left(\begin{array}{c} n-m \\ n\end{array}\right);$ 
\item $\mathbf{t} = \frac{1}{\sqrt{n^2+m^2}} \left(\begin{array}{c} m \\ n \end{array}\right), \quad \boldsymbol{\mu} = \left(\begin{array}{cc} 1 & -1 \\-1 & 0  \end{array}\right), \quad \mathbf{f} \parallel \left(\begin{array}{c} m - n \\ 2m-n\end{array}\right), \quad  \mathbf{g} \parallel \left(\begin{array}{c} m+n \\ -n\end{array}\right);$ 
\item $\mathbf{t} = \frac{1}{\sqrt{n^2+m^2}} \left(\begin{array}{c} m \\ n \end{array}\right), \quad \boldsymbol{\mu} = \left(\begin{array}{cc} 0 & 1 \\1 & 1  \end{array}\right), \quad \mathbf{f} \parallel \left(\begin{array}{c} -m - 2n \\ m+n\end{array}\right), \quad  \mathbf{g} \parallel \left(\begin{array}{c} -m \\ n-m\end{array}\right);$ 
\item $\mathbf{t} = \frac{1}{\sqrt{n^2+m^2}} \left(\begin{array}{c} m \\ n \end{array}\right), \quad \boldsymbol{\mu} = \left(\begin{array}{cc} 0 & -1 \\-1 & 1  \end{array}\right), \quad \mathbf{f} \parallel \left(\begin{array}{c} m - 2n \\ m-n\end{array}\right), \quad  \mathbf{g} \parallel \left(\begin{array}{c} -m \\ n+m\end{array}\right);$ 
\item  for each of the above, replace $\boldsymbol{\mu}$ by $-\boldsymbol{\mu}$ and switch $\mathbf{f}$ and $\mathbf{g}$.
\end{itemize}
\end{enumerate}

\begin{figure*}[ht!]
	\centering
	\includegraphics[width=0.5\textwidth]{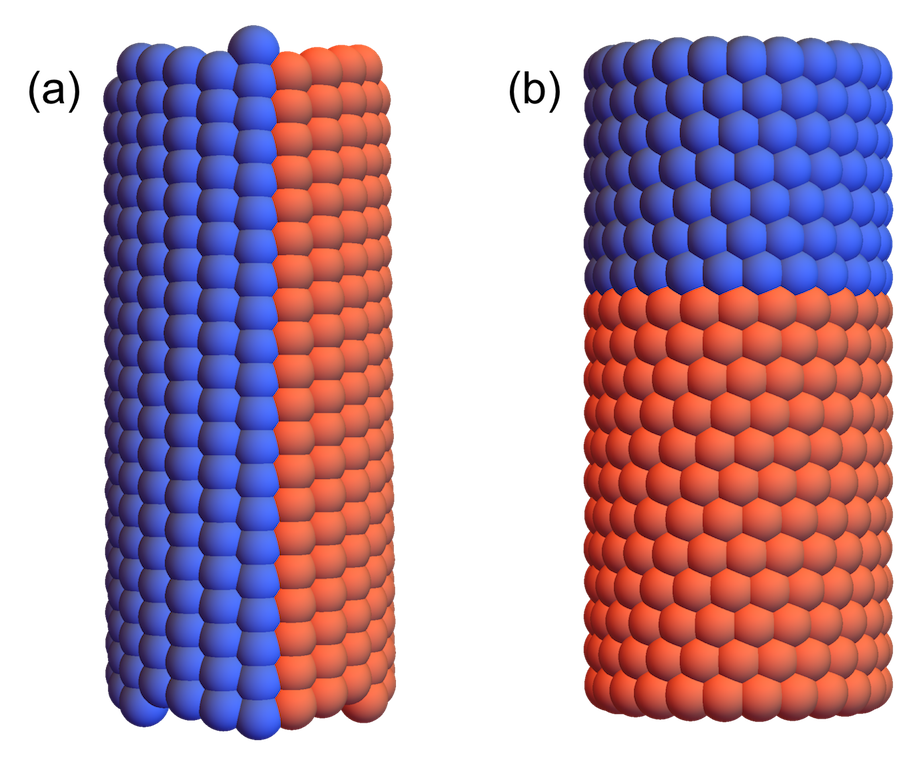}
	\caption{Examples of vertical and horizontal twins.  (a). A vertical twin with parameters  $\mathbf{t} = \mathbf{e}_1$, $\boldsymbol{\mu}= (\mathbf{e}_1, \mathbf{e}_1 -\mathbf{e}_2),$ $\mathbf{f} = (\pi/12) \mathbf{e}_2$, $\mathbf{g} = 0.0625 (4,3)$.  (b). A horizontal twin with parameters $\mathbf{t} = \mathbf{e}_1$, $\boldsymbol{\mu}= (-\mathbf{e}_1, -\mathbf{e}_1 +\mathbf{e}_2),$ $ \mathbf{f} =  (\pi/12)(1,.6)$, $\mathbf{g} = (\pi/12)\mathbf{e}_2$.}
	\label{fig:vertical_horizontal_twin}
\end{figure*}

\begin{figure*}[ht!]
	\centering
	\includegraphics[width=0.8\textwidth]{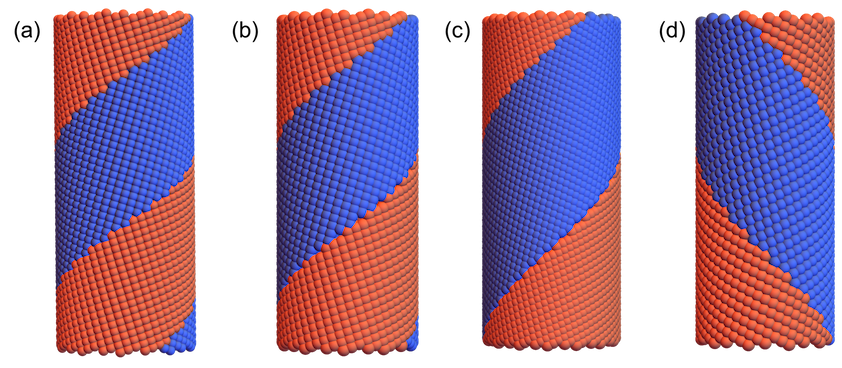}
	\caption{Examples of helical twins. The parameter are: (a). $\mathbf{t} = (1/\sqrt{13})(3,-2)$, $\boldsymbol{\mu} = (\mathbf{e}_1 + \mathbf{e}_2, \mathbf{e}_1)$, $\mathbf{f} = (\pi/101)(-1,4)$, $\mathbf{g} = -0.02(5,2)$;  (b). $\mathbf{t} = (1/\sqrt{13})(3,2)$, $\boldsymbol{\mu} = (\mathbf{e}_1 - \mathbf{e}_2, -\mathbf{e}_1)$, $\mathbf{f} = (\pi/101)(1,4)$, $\mathbf{g} = 0.0225(5,-2)$; (c). $\mathbf{t} = (1/\sqrt{13})(2,-3)$, $\boldsymbol{\mu} = (\mathbf{e}_2,\mathbf{e}_1+  \mathbf{e}_2)$, $\mathbf{f} = (\pi/154)(4,-1)$, $\mathbf{g} = -0.02(2,5)$;  (d). $\mathbf{t} = (1/\sqrt{13})(2,3)$, $\boldsymbol{\mu} = (-\mathbf{e}_2, -\mathbf{e}_1+ \mathbf{e}_2)$, $\mathbf{f} = (\pi/101)(-4,-1)$, $\mathbf{g} = 0.03(-2,5)$.}
	\label{fig:helical_twin}
\end{figure*}

The \textit{twin vector} $\mathbf{t}_{\mathbf{y}}$ on the helical structure is defined analogously to the slip vector in (\ref{eq:slipVector}).  The twin vectors for vertical and horizontal twins are $\mathbf{e}$ and $\mathbf{W}\mathbf{r}$, respectively, but the twin vectors for helical twins can take many forms.   For instance, those corresponding to the examples in Figure  \ref{fig:helical_twin} are given by (in the $\{ |\mathbf{r}|^{-1}\mathbf{r} ,  |\mathbf{r}|^{-1}\mathbf{W}\mathbf{r}, \mathbf{e}\}$ basis):

\begin{table}[!ht]
	\centering
	\begin{threeparttable}
		\begin{tabular}{ccccc}
			\toprule
			Twin Vector &	Helical (a)		  & Helical (b)			    &   Helical (c) & 	Helical (d)	  	 \\
			\midrule 	
			$\mathbf{t}_{\mathbf{y}}$ &  $\left(\begin{array}{c} 0 \\  -0.819123\\  -0.573617\end{array}\right)$ & $\left(\begin{array}{c}0 \\  0.785515 \\ 0.618842\end{array}\right)$ & $\left(\begin{array}{c} 0\\ 0.714072\\ 0.700072\end{array}\right)$ & $\left(\begin{array}{c} 0\\ -0.68951\\ 0.724277\end{array}\right)$  \\
			\bottomrule
		\end{tabular}
	\end{threeparttable}
\end{table}

\begin{proof}[Proof of Lemma \ref{TwinningLemma}.] (i). Necessary and sufficient conditions are  $\mathbf{g} \cdot (\boldsymbol{\mu} \mathbf{t} - \mathbf{t}) = 0$, $\mathbf{g} \cdot \mathbf{t} \neq 0$ and  $\boldsymbol{\mu} \mathbf{t} \cdot \mathbf{f} = -\mathbf{f} \cdot \mathbf{t} = 0$ for some $\boldsymbol{\mu} \in \mathcal{N}^{(-)}(GL(\mathbb{Z}^2))$.  The latter conditions imply $\mathbf{f} \parallel \mathbf{t}^{\perp}$ and $\boldsymbol{\mu} \mathbf{t} = \lambda \mathbf{t}$ for some $\lambda \in \mathbb{R}$.  The two former conditions then imply $\lambda =  1$.  This is necessary and sufficient so long as $\boldsymbol{\mu} \in \mathcal{N}^{(-)}(GL(\mathbb{Z}^2))$.  (ii). This follows by the argument above after reversing the roles of $\mathbf{f}$ and $\mathbf{g}$ and replacing $\boldsymbol{\mu}$ with its minus. (iii). Necessary and sufficient conditions are $\bff \cdot (\boldsymbol{\mu} + \bfI)\bft  =0$, $\bfg \cdot (\boldsymbol{\mu} - \bfI)\bft =0$, $\bff \cdot \bft \ne 0,\  \bfg \cdot \bft \ne 0$ and $\boldsymbol{\mu} \in  \mathcal{N}^{(-)}(GL(\mathbb{Z}^2))$.  There are linearly
independent solutions $\{\bff, \bfg\}$ of these two equations if and only
if $ (\boldsymbol{\mu} + \bfI)\bft$ and $ (\boldsymbol{\mu} - \bfI)\bft$ are not parallel.
Note also that the conditions $\bff \cdot \bft \ne 0,\  \bfg \cdot \bft \ne 0$ imply that $\boldsymbol{\mu} \bft \ne \pm \bft$.  Hence, $ (\boldsymbol{\mu} + \bfI)\bft$ and $ (\boldsymbol{\mu} - \bfI)\bft $ are not parallel if and only if $\bft$ is not an eigenvector of $\boldsymbol{\mu} \in  \mathcal{N}^{(-)}(GL(\mathbb{Z}^2))$. 
\end{proof}

\

\section{Discussion}\label{sec9}
We have developed a theoretical framework for investigating local compatibility between any two helical phases. This involves the description of the original helical group (Section \ref{sec2}), the nearest neighbor reparameterization of the group (Section \ref{sec3} and \ref{sec4}), and the continuous compatibility conditions (Section \ref{sec5}-\ref{sec7}).  Through rigorous justification, we have shown that there are four and only four types of locally compatible interface of a helical structure.  These are vertical, horizontal, helical and elliptical interfaces\textemdash each named for their physical appearance on the structure.  We have specialized these results to the case in which the two phases are the same and we have noticed that additional interfaces are possible when we allow near (as opposed to nearest) neighbor generators.  We then  classified the structural parameters under which near-neighbor-generated interfaces can form in a single phase; these are naturally interpreted as slips and twins (Section \ref{sec8}).  In this section we discuss qualitatively some of the more striking features of helical structures that can be explored with this theoretical framework.

\vspace{2mm}

\noindent \textbf{Mechanical twinning for large and reversible twist.}  If a helical structure has chirality or handedness, then twinning provides a promising mechanism to induce large macroscopic deformation at small elastic stress.  Consider the example in Figure \ref{fig:twist_horizontal}.  The blue phase atoms are arranged so that their is a line of atoms perfectly horizontal on the circumference of the cylinder and another line that loops around the cylinder in a ``right-handed" fashion.  When subject macroscopic twist, the structure can accommodate the twist by deforming uniformly away from its preferred chirality at the cost of significant elastic stresses.  But it has an alternative.  Since its structural parameters satisfy the conditions of compatibility to form a horizontal twin, this structure can twist by the motion of twinned interfaces. Notice in the figure that the introduction of a twinned interface in this structure creates a mirrored (i.e., ``left-handed") red phase, and the motion of this interface results  in a change in volume fraction of the right and left handed phases.  This corresponds to a macroscopic twist.  More than that, if the initial phase is a stress-free equilibrium, then the mirrored phase and the mixtures should also be nearly stress-free (except, perhaps, close to the interface). So, we achieve large twisting deformation at little stress in a process that is (ideally) completely reversible.

\begin{figure}[ht!]
\centering
\includegraphics[width = 6in]{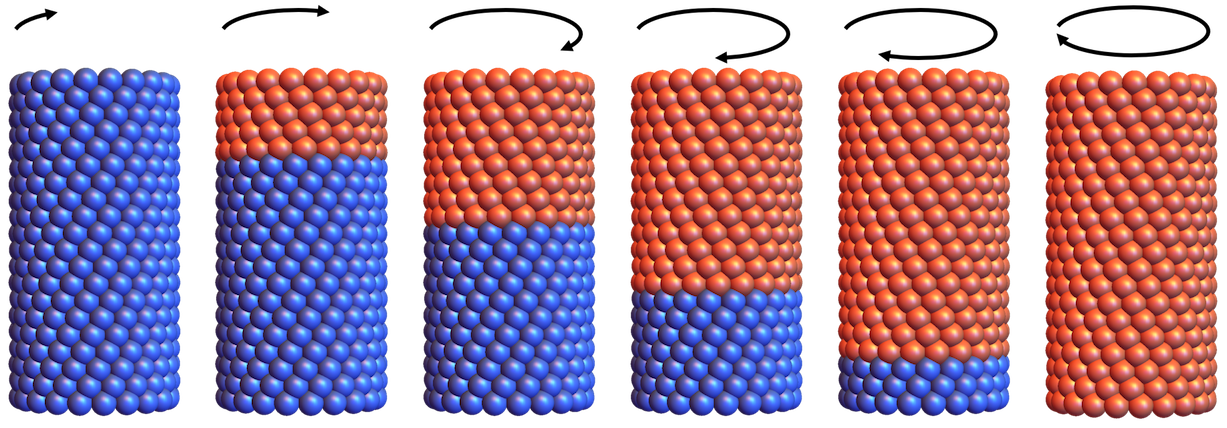}
\caption{Macroscopic twist induced by mechanical twinning at a horizontal interface. The parameters are $\mathbf{t} = \mathbf{e}_2$, $\boldsymbol{\mu}= (\mathbf{e}_1-\mathbf{e}_2, -\mathbf{e}_2),$ $ \mathbf{f} =  (0.12,\pi/10)$, $\mathbf{g} = (-0.264,0)$.}
\label{fig:twist_horizontal}
\end{figure}

\begin{figure}[ht!]
\centering
\includegraphics[width = 6in]{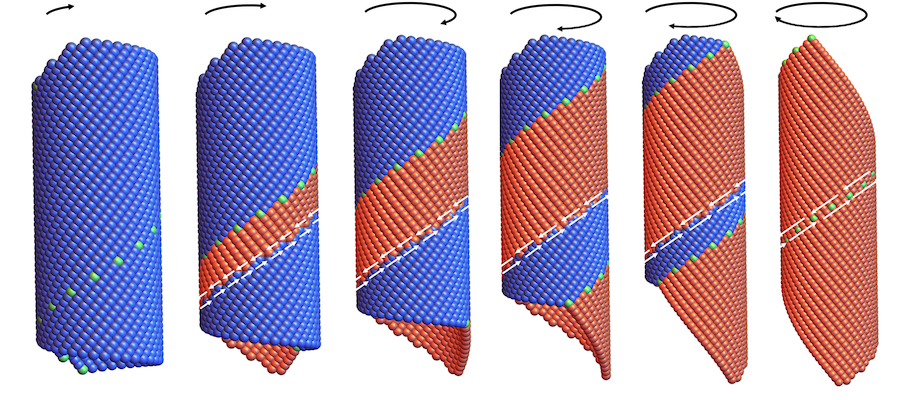}
\caption{Macroscopic twist and extension induced by twinning and slip at a helical interface. The parameters are $\mathbf{t} = (2,3)$, $\boldsymbol{\mu} = (\mathbf{e}_1 + \mathbf{e}_2, \mathbf{e}_1)$, $\mathbf{f} = (\pi/88)(-5,7)$ and $\mathbf{g} = 0.0225(1,3)$.  Points on the locally compatible interface are displayed in green.}
\label{fig:slip_helical}
\end{figure}

\vspace{2mm}

\noindent \textbf{Helical twins as the result of a phase transformation.} The structural parameters that enable a horizontal twin are, unfortunately, quite restrictive.  A horizontal line of atoms along the circumference is required, and this is far from generic\footnote{although it may be induced by a particular macroscopic twist and extension in a structure that does not exhibit this feature in the stress-free state.}.  On the other hand, many choices of parameters  allow for locally compatible helical twins. It is tempting to think a similar correspondence between mechanical twinning and macroscopic deformation applies here, but helical interfaces are subject to a delicate notion of global compatibility.  Consider the helical twin in Figure \ref{fig:slip_helical}.  
\begin{figure}[!ht]
	\centering
	\includegraphics[width = 6in]{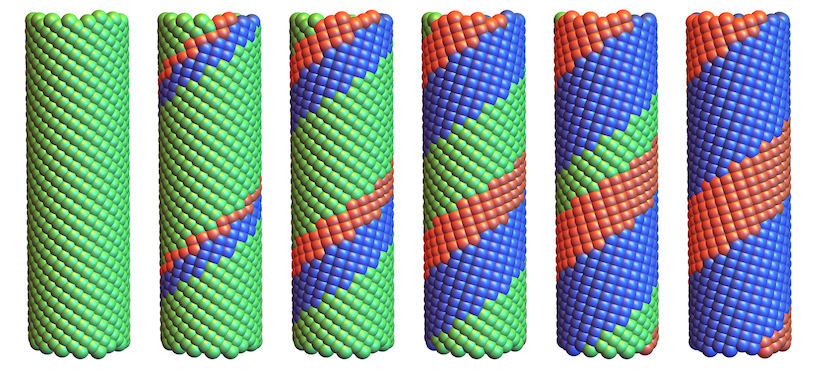}
	\caption{A phase transition results in a helical interface.  The reference tangent $\mathbf{t} = (3,-2)$. The helical twin has parameters $\mathbf{f} = (\pi/55)(-1,4)$, $\mathbf{g} = 0.04(-5,-2)$ and $\boldsymbol{\mu} = (\mathbf{e}_1 + \mathbf{e}_2, \mathbf{e}_1)$.  The parent phase has parameters $(\bar{\mathbf{f}}, \bar{\mathbf{g}}) = \lambda (\mathbf{f}, \mathbf{g}) + (1-\lambda) ( -\boldsymbol{\mu}^T \mathbf{f}, \boldsymbol{\mu}^T \mathbf{g})$ for $\lambda = .6$.  }
	\label{fig:HelicalPhaseChange}
\end{figure}
Here, we resolve the local compatibility condition at the interface associated with the green atoms, so that neighboring atoms prior to transformation remain neighbors across this interface after transformation.  However, helical interfaces come in pairs\textemdash the blue phase is above the red phase for the locally compatible interface, but there is a second interface with the opposite orientation.   The correspondence of atoms across this second interface may look perfect, but any motion involving a change in the volume fraction of the phases, such as the one shown in the figure, {\it necessarily results in slip}.  Thus, mechanical twist in these instances is achieved only through a combination of twinning and slip.  Consequently, the volume fraction of phases for these interfaces is, in a certain sense, topologically protected.

This leads to a fundamental question: Can any of the helical interfaces be achieved without slip?  The answer is yes, and potentially generically so (though, there is still much to discover in this direction).  The setting is that of a phase transformation as shown in Figure \ref{fig:HelicalPhaseChange}. A helical structure, initially of a certain preferred chirality (green), is subject to a stimuli which changes its free energy so that a second chirality (blue and its mirror in red) becomes the preferred state.  The two phases may satisfy the conditions of local compatibility along a helical interface, but they will be unable to coexist as pairs without slip due to the global compatibility condition discussed above.  However, by twinning the second phase at exactly the right volume fraction\textemdash so that it accommodates the chirality of the first phase\textemdash we can achieve a globally compatible helical structure that involves no slip and no elastic stresses, all while cycling back and forth between phases.  This is a fascinating analog of ideas of self-accommodation and  $\lambda_2 = 1$ for crystalline solids undergoing phase transformations   \cite{bhattacharya_self-accommodation_1992, chen_study_2013,gu_17}.

\vspace{2mm} 

\noindent \textbf{Vertical interfaces in microtubules.}  As a final comment to
\begin{figure}[!ht]
	\centering
	\includegraphics[width = 6.5in]{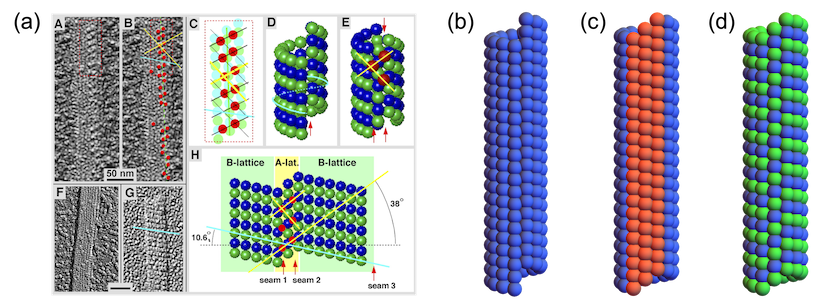}
	\caption{Vertical interfaces in a microtubule. (a).  The microtubule lattice seam and compatible vertical interfaces (Reproduced with permission \cite{SANDBLAD20061415} \textcopyright 2006 Elsevier Inc.).  (b-d) A 13 protofilaments microtubule constructed by our method.  (b) The lattice seam created by wrapping the ``B" lattice to form a tube.  (c-d). The microtubule as the combination of ``A" and ``B" phases.  These form compatible vertical interfaces.  The parameters are:  (Blue) $\bff_a=(2\pi/13, 0)$, $\bfg_a=(\frac{2\pi}{13}\tan(\frac{-10.6\pi}{180}), 2 \pi/13)$. (Red) $\bff_b=(2\pi/13, 0)$, $\bfg_b=(\frac{2\pi}{13}\tan(\frac{38\pi}{180}), 2\pi/13)$ and $\boldsymbol{\mu} = \mathbf{I}$. The reference tangent is $\bft=(0,1)$.   }
	\label{fig:mirotubule}
\end{figure}
emphasize the importance of compatible interfaces in helical structures, we introduce a biological example: the fission yeast end binding protein (EB1) homolog Mal3p in microtubules \cite{SANDBLAD20061415}.  Microtubules are dynamic tubular structures involved in intracellular transport, flagellar motion, and many other tasks in cells.  Their structure is often that of a non-discrete helical group.  Specifically, in the example shown in Figure \ref{fig:mirotubule}(a), the preferred ``B" lattice is described by the angle $10.6^{\text{o}}$ when unrolled.  However, wrapping this ``B'' lattice onto a tubular structure leads to a large vertical seam on the tube \ref{fig:mirotubule}(b).   \textit{This is non-discreteness.}  The authors in \cite{SANDBLAD20061415} argue that this seam, in particular, is much too large for the microtubule to be stable in this structure.  Instead, the microtubule forms a vertical strip corresponding to the ``A" lattice depicted in \ref{fig:mirotubule}(a).  This has the effect of stabilizing the structure by making the seam more coherent.   In the context of our theoretical framework, this mechanism is exactly that of compatible vertical interfaces between  the two phases (Figure \ref{fig:mirotubule}(c-d)).

Besides the microtubule, there are many possible applications of these results above to nanotubes, inorganic or biological.  In this paper we have concentrated on the basic theory, especially the classification of {\it all possible} compatible interfaces and their mobility.  In forthcoming work
we will apply this theory to interesting special cases, especially to guide the design of structure-dependent tension-twist protocols that can induce 
specific phase transformations.

\vspace{5mm}

\noindent {\bf Acknowledgment.} This work was supported by the MURI program (FA9550-18-1-0095, FA9550-16-1-0566) and ONR (N00014-18-1-2766).  It also benefitted from the support of NSF (DMREF-1629026),  the Medtronic Corp, the Institute on the Environment (RDF fund), and the Norwegian Centennial Chair Program.
\bibliographystyle{plain}
\bibliography{transformations_reference}

\end{document}